\documentclass[]{interact}
\usepackage{float}
\usepackage{xcolor,diagbox}
\usepackage[bookmarks]{hyperref}
\usepackage{stmaryrd}
\usepackage{subcaption}
\allowdisplaybreaks
\usepackage{gn-logic14,pifont}
\usepackage{amssymb,amsthm,latexsym}
\usepackage{framed}

\definecolor{DarkGreen}{RGB}{10,50,2}
\newcommand{\new}[1]{#1}
\newcommand{\newnew}[1]{#1}

\newcommand{\op}{\mathrm{opt}}

\newcommand{\mx}{\mathrm{max}}

\usepackage{listings}
\usepackage[natbibapa,nodoi]{apacite}
\renewcommand\bibliographytypesize{\fontsize{10}{12}\selectfont}

\makeatletter
\usepackage{etoolbox}
\makeatletter
\pretocmd{\@cite}{\def\@BBOP{[}\def\@BBCP{]}}{}{}
\makeatother

\theoremstyle{plain}
\newtheorem{theorem}{Theorem}[section]

\newtheorem{corollary}[theorem]{Corollary}
\newtheorem{proposition}[theorem]{Proposition}

\theoremstyle{definition}

\theoremstyle{remark}
\newtheorem{remark}{Remark}

\usepackage[disable]{endfloat}

\begin{document}


\title{
Conditional normative reasoning as a fragment of HOL 
}
\newcommand{\xav}[1]{{\color{blue}XAV #1 }}
\author{
\name{Xavier Parent\textsuperscript{a}\thanks{CONTACT X. Parent. Email: xavier.parent@tuwien.ac.at; C. Benzmüller: christoph.benzmueller@uni-bamberg.de. This article will be published in the Journal of Applied Non-Classical Logics, 2024.} and Christoph Benzmüller\textsuperscript{b} 
\textsuperscript{c}
}
\affil{\textsuperscript{a}
Technische Universität Wien, Favoritenstrasse 9,
A-1040 Wien, Austria
\textsuperscript{b} Universität Bamberg, An der Weberei 5, 96047 Bamberg, Germany
\textsuperscript{c} Freie Universität Berlin, Arnimallee 7, 14195 Berlin, Germany
}
}
\maketitle

\begin{abstract} 
We report on the mechanization of (preference-based) conditional normative reasoning. Our focus is on \AA qvist's system {\bf E} for conditional obligation, \new{and its extensions}. Our mechanization is achieved via a shallow semantical embedding in Isabelle/HOL.  We consider two possible uses of the framework. The first one is as a tool for meta-reasoning about the considered logic. We employ it for the 
automated verification of deontic correspondences  (broadly conceived) and related matters, analogous to what has been previously achieved for the modal logic cube.  \newnew{The equivalence is automatically verified in one direction, leading from the property to the axiom.} The second use is as a tool for assessing ethical arguments. We provide a computer encoding of a well-known paradox \newnew{(or impossibility theorem)} in population ethics, Parfit's repugnant conclusion. 
\newnew{While some have proposed overcoming the impossibility theorem by abandoning the presupposed transitivity of ``better than,"  our formalisation unveils a less extreme approach, suggesting among other things the option of weakening transitivity suitably rather than discarding it entirely.}
Whether the presented encoding increases or decreases the attractiveness and persuasiveness of the repugnant conclusion is a question we would like to pass on to philosophy and ethics.


\end{abstract}



\section{Introduction} \label{intro}

We report on the mechanization of (preference-based) conditional  normative reasoning. Our focus is on \AA qvist's system {\bf E} for conditional obligation, \new{and its extensions}. Our mechanization is achieved via a shallow semantical embedding in Isabelle/HOL adapting the methods used by \cite{C47}.
To look at Standard Deontic Logic (SDL) and extensions \citep{ddl:C80,PvTT21} would not be very interesting. First, no new insights would be gained, since SDL is a normal modal logic of type KD, which is already covered by the prior work of Benzmüller and colleagues. \new{Secondly}, SDL is vulnerable to the well-known deontic paradoxes, including  in particular Chisholm's paradox of contrary-to-duty obligation, see \cite{PvTT21} for details. We thus focus here on  Dyadic Deontic Logics (DDLs) with a preference-based semantics, which originate from the works of  \cite{ddl:H69} and \cite{ddl:L73}. 
To represent conditional obligation sentences, an ``intensional'' dyadic operator (which is weaker than material implication) is employed. The proposed semantics generalizes that of SDL: the SDL-ish binary classification of states into good/bad is relaxed to allow for grades of ideality and accommodate classifications such as best, 2nd-best, and so forth.  
More specifically, a preference relation $\succeq$  ranks the possible worlds in terms of comparative goodness or betterness.\footnote{For $i\succeq j$, read ``$i$ is at least as good as $j$''. } The conditional obligation of $\psi$, given $\varphi$ (notation: $\bigcirc (\psi/\varphi)$) is evaluated as true if the best $\varphi$-worlds are all $\psi$-worlds. 
Like in modal logic, different properties of the betterness relation yield different systems. For further details on this framework, the reader is referred to the overview chapter by \cite{ddl:P21} 
found in the second volume of the \emph{Handbook of Deontic Logic and Normative Systems}.

In this paper, our emphasis is on two possible uses of the mechanized tool. First, we employ it as a tool for meta-reasoning about the considered logics.
So far the correspondences between properties and modal axioms have been established ``with pen and paper". This raises the question of how much of these correspondences can be automatically explored by modern theorem-proving technology. 
The automatic verification of correspondences can be done for the modal cube \citep{C47}. We want to understand if it can also be done for DDL.
\cite{C47} write: ``automation facilities could be very
useful for the exploration of the meta-theory of other logics, for example, conditional logics, since
the overall methodology is obviously transferable to other logics of interest". 
Here we follow up on that suggestion, building on further prior results from \cite{J45}, where the weakest available system (called  {\bf E}) has faithfully been embedded in Higher-Order Logic (HOL). In the present paper, we consider extensions of {\bf E}. 
We look at connections or correspondences between axioms and semantic conditions as ``extracted" by relevant soundness and completeness theorems. 
Thus, ``correspondence" is taken in the same (broad) sense that Hughes and Cresswell have in mind when they write: 
\begin{quote}
    ``D, T, K4, KB [are] produced by adding a single axiom to K and [...] in each case the system turns out to be characterized by  [sound and complete wrt] the class of models in which [the accessibility relation] $R$ satisfies a certain condition. When such a situation obtains--i.e. 
    when a system K+$\alpha$ is characterized by the class of all models in which $R$ satisfies a certain condition$-$we shall [...] say [...] that the wff $\alpha$ itself is characterized by that condition, or that the condition \emph{corresponds} [their italics] to $\alpha$." \cite[p. 41]{HC84}
    \end{quote}

\new{This is different from correspondence theory in the sense of \cite{Sah75} and \cite{Ben2001}. 
Typically, Sahlqvist-style modal correspondence theory studies the equivalence between modal formulas and first-order formulas over Kripke frames
via the so-called standard translation. The goal is to identify syntactic classes  of modal formulas that can be shown to define first-order conditions on frames, and which are themselves computable via an algorithm. Correspondence theory in this sense has not been developed for preference-based dyadic deontic logic and conditional logic yet. This is in part due to the more complex form of the truth conditions for the conditional.  The Sahlqvist/van Benthem method allows to establish an equivalence between an axiom and a property. By contrast, our method will give us only one direction of the equivalence, from the property to the axiom, but not yet the other direction. 

A distinctive feature of our method is its flexibility. We will primarily deal with the conventional evaluation pattern in terms of best antecedent-worlds, distinguishing between two notions of best, optimality and maximality, as is the custom in rational choice theory. For the sake of completeness, we will also consider variant truth conditions that do not rely on the limit assumption, which some consider controversial. Notably, \cite{ddl:L73} rejected this assumption. In a deontic context, it amounts to assuming the existence of ``the best of all possible worlds''.  For simplicity, we will confine our analysis to Lewis's variant rule. 

}

The second use we consider for our mechanized system is as a tool for assessing ethical arguments in philosophical debates. As an illustration, we look at one of the well-known paradoxes or impossibility theorems in population ethics, the so-called ``repugnant conclusion" due to \cite{Parfit1984-PARRAP}. We provide a computer encoding of the
repugnant conclusion to make it amenable to formal analysis and computer-assisted experiments. We believe that the formalization has the potential to further stimulate the philosophical debate on the repugnant conclusion, given the considerable simplifications it achieves.
Specifically, our formalization hints at the possibility of a fresh perspective on the scenario. While some have proposed overcoming the impossibility theorem by relinquishing the presupposed transitivity of ``better than," this solution is often deemed too radical. \newnew{We distinguish between ``better than'' as a relation on formulas and as a relation on possible worlds, the second being
used to elucidate the formal meaning of the first. Shifting the emphasis on the second},  our formalisation unveils a less extreme approach. It consists in  weakening transitivity suitably rather than discarding it entirely.
\newnew{However, we show that not all candidate weakenings of transitivity will do. In particular, drawing on \cite{ddl:parent24}, we argue that acyclicity (or even quasi-transitivity) does the job, but not the interval order condition. We also raise the question if transitivity is the sole cause of the paradox. We point out that under the standard interpretation of ``best'' in terms of maximality (quasi-)transitivity generates an inconsistency only if the set of possible worlds is assumed to be finite$-$an assumption  that might appear overly limiting, if not arbitrary.  This finding allows us to resolve (negatively) an open problem from previous work:\footnote{Cf. \cite{ddl:lou19} and Parent (\citeyear{ddl:P21,ddl:parent24}).} whether under the rule of maximality the finite model property holds for preference models with transitive, quasi-transitive, or interval order relations.

Until now, these particular points have remained unnoticed. Previously, one could verify them manually, but now, automation eliminates the need for logical expertise. Additionally, experimenting and implementing variations, such as changing the evaluation rule for the conditional, is straightforward.
The practicality of the proposed tool lies in its ability to swiftly (dis-)confirm alternative hypotheses with minimal reliance on logical expertise.
With this case study--the first of its kind--we hope to provide evidence that automated tools may help to facilitate the understanding and assessment of ethical arguments in philosophical debates.
Previously, the second author utilized similar techniques in computational metaphysics. Notably, the inconsistency within the axioms of Gödel's ontological argument went unnoticed until 2013, when it was automatically identified by the higher-order theorem prover Leo-II$-$see \cite{C55}.}

\new{Readers should be warned that there is less standardization in preference semantics for dyadic deontic logic
than in the usual Kripke-style semantics for (monadic) deontic logic, and more room for variation.
This is because several factors must be juggled all at once. In
this paper we stick to \cite{ddl:A87,ddl:A02}'s approach, but the account is also applicable to further variants.  Those who wish to get a general overview of
the possible approaches that can be taken might find it useful to consult \cite{ddl:M93}.
The interested reader will find in \cite{ddl:lou19} and \cite{ddl:P21} further pointers to the literature.}



The paper is organized as follows. Section \ref{e} recalls system  {\bf E} and its extensions.
 Section \ref{hol} shows the embedding of {\bf E} in Isabelle/HOL. 
Section \ref{cor} studies the correspondence between the properties of the betterness relation and the axioms. 
Section~\ref{rc} discusses the repugnant conclusion. Section~\ref{conc} concludes.\footnote{The theory files are available for downloading  at \url{http://logikey.org} under sub-repository ``/Deontic-Logics/cube-dll/" (files ``\new{DDLcube.thy", ``mere\_addition\_opt.thy", ``mere\_addition\_max.thy'' and ``mere\_addition\_lewis.thy''}). 
A corresponding (but slightly modified) Isabelle/HOL dataset is presented in \cite{CondNormReasHOL-AFP}.}



\section{System {\bf E}} \label{e}
We describe the language, the semantics and proof theory of system {\bf E} and its extensions. 

\new{\subsection{Language}

The language, call it  $\mathcal{L}$,
is defined by the following BNF:
\begin{flalign*}
& \mbox{Atomic formulas: } \phantom{x}p\in \mathbb{P} \\
& \mbox{Formulas: }\phantom{xxxxxx} A\in \mathcal{L} 
\\ A::= p\mid &\neg A \mid A\vee A\mid \square A\mid
	\bigcirc (A/A)
\end{flalign*}
$\neg A$ is read as ``not-$A$'', and $A\vee B$ as ``$A$ or
$B$''. $\square A$ is read as ``$A$ is settled as true'', and
$\bigcirc(B/A)$ as ``$B$ is obligatory, given $A$''. 

The Boolean connectives other than ``$\neg$" and ``$\vee$" are defined as usual. $\Diamond A$ is short for $\neg\square\neg
A$. $P(B/A)$ (``$B$ is
permitted, given $A$'') is short for $\neg \bigcirc(\neg B/A)$,
$\bigcirc A$ (``$A$ is unconditionally obligatory'') and $PA$ (``$A$ is
unconditionally permitted'') are short for $\bigcirc(A/\top)$ and
$P(A/\top)$, where $\top$ denotes a tautology.

}




\subsection{Semantics}
We start with the main ingredients of the semantics. A preference model is a structure
$M=(W, \succeq, v)$, where $W$ is a non-empty set of possible worlds (called its ``universe''), $\succeq$ is a preference relation ranking the elements of $W$ in terms of betterness or comparative goodness, and $v$ is a function assigning to each atomic formulas a subset of $W$ (intuitively, the subset of those worlds where the atomic formula is true).
$a\succeq b$ may be read ``$a$ is at least as good as $b$". \new{Also,} $\succ$ is the strict counterpart of $\succeq$, defined by $a\succ b$ ($a$ is strictly better than $b$) iff $a\succeq b$ and
$b\not\succeq a$. \new{And} $\approx$ is the equal goodness relation, defined by $a\approx b$ ($a$ and $b$ are equally good)  iff $a\succeq b$ and
$b\succeq a$. \new{For future reference, note that by definition $\succ$ is irreflexive (for all $a$, $a\not\succ a$) and asymmetric (for all $a$, $b$, if $a\succ b$ then $b\not\succ a$).}

A model $M$ is said to be finite if its universe $W$ is. The truth conditions for modal and deontic formulas read:
\begin{itemize}
	\item $M,a\vDash \Box \varphi
   \mbox{ iff }  \forall b \in W \mbox{ we have }  M,b\vDash \varphi$ 
	\item $M,a\vDash \bigcirc (\psi/\varphi) \mbox{ iff } \forall b\in \mathrm{best}(\varphi) \mbox{ we have }	M, b\vDash \psi$
\end{itemize}
When no confusion can arise, we omit the reference to $M$ and simply write $a\models \varphi$. Intuitively, $\bigcirc (\psi/\varphi)$ is  true if the  best $\varphi$-worlds are all $\psi$-worlds. There is variation among authors regarding the formal definition of ``best''. It is sometimes cast in terms of maximality (we call this the max rule) and some other times cast in terms of optimality (we call this the opt rule).
\new{A} $\varphi$-world $a$ is maximal if it is not (strictly) worse than any other $\varphi$-world. It is optimal if it is at least as good as any $\varphi$-world. \new{An optimal element is maximal, but not the other way around.} The two notions coincide only when ``gaps" (incomparabilities)  in the ranking are ruled out. Formally: 

	\begin{center} 
 $\begin{array}{|c|c|}
\hline
\textbf{Max rule} & \textbf{Opt rule} \\
	\mathrm{best}(\varphi) =\mx(\varphi) &	\mathrm{best}(\varphi)=\op(\varphi)\\
\hline
\end{array}$
 \end{center}
where
	\begin{flalign*}
			a\in \mx(\varphi) 
		& \Leftrightarrow  a\models 
		\varphi\;\&\; \neg \exists b \;(b \vDash \varphi \;\&\; b\succ a) \\
		a\in \op(\varphi) &\Leftrightarrow  a\models 
		\varphi\;\&\; \forall b \;(b \vDash \varphi
		\rightarrow a\succeq b)
		\end{flalign*}
The relevant properties of $\succeq$ are (universal quantification over worlds is left implicit):

\begin{itemize}
    \item Reflexivity: $a\succeq a$;
        \item \new{Totality} or (strong) connectedness:  $a\succeq b$ or  $b\succeq a$ (or both);
        \item Transitivity: if $a\succeq b$ and  $b\succeq c$, then  $a\succeq c$;
        \item \new{Various weakenings of transitivity (from so-called rational choice theory):}
        \begin{itemize}
    \item \new{Quasi-transitivity: if $a\succ b$ and $b\succ c$ then $a\succ c$;
   \item Acyclicity: if 
$a\succ^{\star}b$, then
$b\not\succ a$, where $\succ^{\star}$ is the transitive closure of $\succ$;
\item Suzumura consistency: 
if 
$a\succeq^{\star}b$, then
$b\not\succ a$, where $\succeq^{\star}$ is the transitive closure of $\succeq$;
 \item Interval order: $\succeq$ is reflexive and Ferrers (if $a\succeq b$ and  $c\succeq d$, then  $a\succeq d$ or $c\succeq b$).}
\end{itemize}
\end{itemize}
\new{Intuitively, quasi-transitivity demands that the strict part of the betterness relation be transitive. Acyclicity rules out strict betterness cycles. Suzumura consistency rules out cycles with at least one instance of strict betterness. 
Acyclicity 
may be interpreted as generalizing asymmetry  to a path of arbitrary length. 
Totality implies reflexivity. Given reflexivity and Ferrers, totality follows, and so the interval order condition can equivalently be defined by the pair  ``totality + Ferres".}
Intuitively, the interval order condition permits instances where transitivity of equal goodness fails, due to discrimination thresholds. These are cases where $a\approx b$ and $b\approx c$ but $a\not\approx c$ (see Luce~\citeyear{L56}).

\new{\begin{figure}[h]\centering\caption{\new{Weakenings of transitivity (\cite{ddl:parent24})}}
\label{wt}
\includegraphics[scale=1]{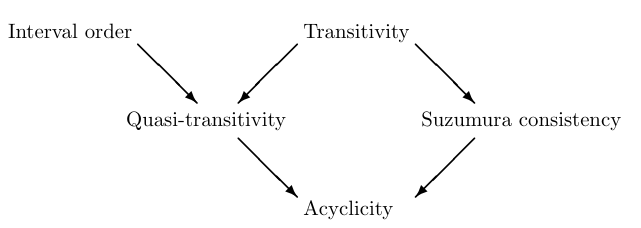}
\end{figure}
These weakenings of transitivity are discussed in greater depth in \cite{ddl:parent24}. Fig.~\ref{wt} 
 shows their relationships. An arrow from one condition to the other means that the first implies the second. The lack of an arrow between two conditions means that they are independent.\footnote{\new{Proofs and additional discussion may be found in \cite{ddl:parent24}.}} 

}

\newpage
Lewis's limit assumption is meant to rule out sets of worlds without a ``limit" (viz. a best element).  Its exact formulation varies among authors. It exists in (at least) the following four versions, where 
$\mathrm{best}\in\{\mx,\op\}$:
\begin{subequations}
\begin{align}
& \mbox{\underline{Limitedness}}\notag\\
&\mbox{ If } \exists x \mbox{ s.t. } x\models\varphi \mbox{ then }\mathrm{best}(\varphi ) \not=\emptyset \tag{LIM} \label{l} \\
& \mbox{\underline{Smoothness} (or stopperedness)}\notag\\
& \mbox{ If } x\models \varphi , \mbox{ then: either } 
x\in \mathrm{best}(\varphi ) \mbox{ or } \exists y \mbox{ s.t. }  
y\succ x \;\& \;y\in\mathrm{best}(\varphi) \tag{SM}\label{smooth}
\end{align}
\end{subequations}
A betterness relation $\succeq$ will be called ``opt-limited" or ``max-limited" depending on whether (\ref{l}) holds with respect to $\op$ or $\mx$. Similarly, it will be called ``opt-smooth"
or ``max-smooth" depending on whether (\ref{smooth}) holds with respect to
 $\op$ or $\mx$. For pointers to the literature, and the relationships between these versions of the limit assumption, see
 \cite{ddl:parent14}.


The above semantics may be viewed as a special case of the selection
function semantics favored 
by Stalnaker and generalized by \cite{ddl:ch75}. The preference relation is replaced with a selection function $f$ from formulas to subsets of $W$, such that, for all $\varphi$, $f(\varphi)\subseteq W$. Intuitively, $f(\varphi)$ outputs all the best $\varphi$-worlds. The evaluation rule for the dyadic obligation operator is thus given as: $\bigcirc (\psi/\varphi)$ holds when $f(\varphi)\subseteq \Vert\psi\Vert$, where $\Vert\psi\Vert$ is the set of $\psi$-worlds. It is known that when suitable constraints are put on the selection function, the two semantics validate exactly the same set of formulas \new{$-$ cf.}  \cite{ddl:parent15} for details.\footnote{One can go one step further, and make the selection function semantics an instance of a more general semantics equipped with a neighborhood function, like in traditional modal logic (cf. \cite{ddl:ch75}). \new{Neighborhood semantics for dyadic deontic logic are investigated by \cite{Segerberg1971}, \cite{nor86} and  \cite{gob04} among others.}} The correspondence between constraints put on the selection function and modal axioms have been verified by automated means by \cite{J26}. A comparison between this prior study and ours is left as a topic for future research. 





\subsection{Systems}

The relevant systems are shown in Fig. \ref{sys}.  A line between two systems indicates that the system to the left is strictly included in the system to the right. \new{{\bf E}, {\bf F} and {\bf G} are from \cite{ddl:A87}. {\bf F}+\ref{cm} and {\bf F}+\ref{dr} are from \cite{ddl:parent14} and \cite{ddl:parent24}, respectively}. 

\begin{center}
\begin{figure}[ht]
\centering
\includegraphics[scale=0.4]{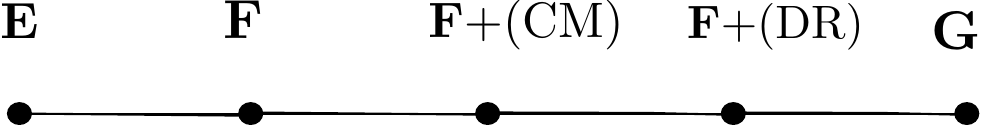}\caption{Systems}\label{sys}
\end{figure}
\end{center}

All the systems contain the classical propositional calculus and the modal system S5.\footnote{S5 is characterized by the rule of necessitation (``If $\vdash A$, then $\vdash\Box A$"), and the K, T and 5 axioms (5 is $\Diamond A\rightarrow\Diamond\Box A$).}  Then they add the following axiom schemata:

\begin{itemize}
\item For {\bf E}  (the naming follows \cite{ddl:P21}):

\vspace{-0.4cm}
\begin{minipage}{10cm}\setlength{\parindent}{2em}
\begin{flalign}
& \mbox{S5-schemata for }\Box
	\tag{S5}\label{a5}\\
	&\bigcirc (\psi\IMPL \xi/\varphi ) \IMPL (\bigcirc (\psi/\varphi) \IMPL \bigcirc
	(\xi/\varphi)\new{)}\tag{COK}\label{cok} \\ &\bigcirc (\psi/\varphi) \IMPL \Box\bigcirc
	(\psi/\varphi) \tag{Abs}\label{abs} \\ & \Box  \varphi\IMPL \bigcirc (\varphi/\psi)
	\tag{Nec}\label{nec} \\ & \Box (\varphi\IFF \psi) \IMPL (\bigcirc (\xi/\varphi)\IFF
	\bigcirc (\xi/\psi) )\tag{Ext}\label{ext} \\ & \bigcirc (\varphi/\varphi)
	\tag{Id}\label{id} \\ & \bigcirc (\xi/\varphi\AND \psi) \IMPL \bigcirc (\psi\IMPL
\xi/\varphi)\tag{Sh}\label{sh} 
\end{flalign}
\end{minipage}
\medskip
\item For {\bf F}: axioms of {\bf E} plus

\vspace{-0.4cm}
\begin{minipage}{10cm}\setlength{\parindent}{2em}
\begin{flalign}
		& \Diamond \varphi\rightarrow (\bigcirc (\psi/\varphi) \rightarrow
		P(\psi/\varphi))\tag{D$^\star$}\label{cod}
\end{flalign}

\end{minipage}

\medskip
\item For {\bf F}+(CM): axioms of {\bf F} plus

\vspace{-0.4cm}
\begin{minipage}{10cm}\setlength{\parindent}{2em}
\begin{flalign}
	&( \bigcirc (\psi/\varphi)\wedge \bigcirc (\xi/\varphi)) \rightarrow \bigcirc
	(\xi/\varphi\wedge \psi)\tag{CM}\label{cm} 
\end{flalign} 
\end{minipage}

\medskip
\item For {\bf F}+(DR): axioms of {\bf F} plus

\vspace{-0.4cm}
\begin{minipage}{10cm}\setlength{\parindent}{2em}
\begin{flalign}
		& \bigcirc (\xi/\varphi\vee \psi) \rightarrow (\bigcirc (\xi/\varphi) \vee \bigcirc (\xi/\psi) )\tag{DR}\label{dr}			
\end{flalign}
\end{minipage}

\medskip
\item For {\bf G}: axioms of {\bf F} plus: 

\vspace{-0.4cm}
\begin{minipage}{10cm}\setlength{\parindent}{2em}
\begin{flalign}
	 & (P(\psi/\varphi)\wedge \bigcirc
		(\psi\rightarrow \xi/\varphi)) \rightarrow \bigcirc (\xi/\varphi\wedge
		\psi)\tag{Sp}\label{spohn}
	\end{flalign} 
\end{minipage}
\end{itemize}
We give an intuitive explanation for these axioms. \ref{cok} is the conditional analog of the familiar distribution
axiom K. \ref{abs} is the absoluteness axiom of~\cite{ddl:L73}, and
reflects the fact that the ranking is not world-relative.  \new{\ref{nec}
is the dyadic deontic counterpart of the familiar necessitation
rule.} \ref{ext} permits the replacement of necessarily equivalent
formulas in the antecedent of deontic conditionals.  \ref{id} is
the deontic analog of the identity principle. \new{ \ref{sh} is named
after~\cite[p.\,77]{ddl:S88}, who seems to have been the
first to discuss it. One can see it as the deontic analog of one-half of the deduction theorem. \ref{cod} is the conditional analog of the familiar D axiom. In its equivalent form, $\Diamond \varphi\rightarrow \neg (\bigcirc (\psi/\varphi) \wedge
		\bigcirc (\neg \psi/\varphi)) $, this axiom rules out the possibility of conflicts between obligations arising in a context $\varphi$ that is possible. }\ref{cm} and \ref{dr} correspond to the principle of cautious monotony and disjunctive rationality from the non-monotonic logic literature \new{$-$\cite{ddl:KLM90}}.
 \ref{cm} tells us that complying with an obligation does not modify the other obligations arising in the same context.  \ref{dr} tells us that if a disjunction of states of affairs triggers an obligation, then at least one disjunct triggers this obligation.  
 \new{ Due to \cite{ddl:S75}, \ref{spohn} is best explained using the (more widely known) principle of rational monotony \ref{rm} from non-monotonic logic$-$see \cite{ddl:KLM90}. The two laws are inter-derivable in {\bf E}.  \ref{rm} is obtained by replacing, in \ref{spohn}, $\bigcirc (\psi\rightarrow \xi/\varphi)$ with $\bigcirc (\xi/\varphi)$, to read:
 \begin{flalign}
	 & (P(\psi/\varphi)\wedge \bigcirc
		(\xi/\varphi)) \rightarrow \bigcirc (\xi/\varphi\wedge
		\psi)\tag{RM}\label{rm}
	\end{flalign} }\ref{rm} says that realizing a permission does not modify our other obligations arising in the same context.

We give below the main soundness and completeness theorems. Those stated in Th.~\ref{sound} hold under both the opt rule and the max rule. It is understood that limitedness is cast in terms of opt when the opt rule is applied, and in terms of max when the max rule is applied. The same holds for smoothness.  
\begin{theorem}[Soundness and completeness,~Parent (\citeyear{ddl:P21,ddl:parent24})]\label{sound} (i)
\label{e:com}{\bf E} is sound and complete w.r.t.~the class of all preference models; (ii)
{\bf F} is sound and complete w.r.t.~the class of preference models in which $\succeq$ is limited; (iii) {\bf F}+\ref{cm} is sound and complete w.r.t.~the class of preference models in which $\succeq$ is smooth; (iv) {\bf F}+\ref{dr} is (weakly) sound and complete w.r.t.~the class of (finite) preference models in which $\succeq$ meets the interval order condition.
\end{theorem}
In part (i), (ii) and (iii) of Th. \ref{sound}, and in Th. \ref{g:sound}, soundness and completeness are taken in their strong sense. They establish a
correspondence between the syntactic and semantic consequence relation while
also accommodating a potentially infinite set of assumptions. To be more precise, the theorems are of the form: where $\Gamma$ is a set of formulas or assumptions,  $\Gamma\vdash\varphi$ if and only if $\Gamma\models\varphi$.  
In part (iv) of Th.~\ref{sound}, soundness and completeness are taken in their weak sense: $\Gamma$ is required to be finite; this amounts to establishing a match between theorems and validities only. This restriction is because the completeness proof appeals in an essential way to the assumption that models are finite$-$for more details, see \cite[\S4]{ddl:parent24}.

 \begin{theorem}[Soundness and completeness, \cite{ddl:parent14}]\label{g:sound} (i)
Under the opt rule
{\bf G} is sound and complete w.r.t.~the class of preference models in which $\succeq$ is limited and transitive; (ii) under the max rule, {\bf G} is sound and complete w.r.t.~the class of preference models in which $\succeq$ is limited, transitive and total. 
\end{theorem}

For more background on these systems, and additional results, we refer the reader to  Parent (\citeyear{ddl:P21,ddl:parent24}).

\subsection{\new{Correspondences}}
Table \ref{results} shows some of the known ``correspondences" between semantic properties and formulas as extracted from Th. \ref{sound} and Th. \ref{g:sound}. \new{Thus, the term ``correspondence" is understood along the lines suggested by Hughes and Cresswell (Cf. Section \ref{intro})}. The leftmost column shows the properties of $\succeq$. The two middle columns show the corresponding modal axioms, the first column for the max rule, and the second one for the opt rule. It is understood that smoothness (resp. limitedness) is defined for max in the max column, and for opt in the opt column. The rightmost column gives the paper where the completeness theorem is established.
The symbol $\times$ indicates that the property (or pair of properties) is known not to correspond to any axiom, in the sense that the property does not modify the set of valid formulas. 

\begin{table}[ht]
\begin{center}
\begin{tabular}{l||c|c|c}
     Property     & Axiom ($\mx$) & Axiom ($\op$) & Ref.\\
\hline\hline 
   reflexivity  & $\times$ & $\times$ &\cite{ddl:parent15} \\
   \newnew{totality}  &   $\times$ & $\times$ &\cite{ddl:parent15}\\
     limitedness    & \ref{cod}  &\ref{cod} & \cite{ddl:parent15}\\
smoothness &  \ref{cm}& \ref{cm} &\cite{ddl:parent14} \\
transitivity& $\times$ & \ref{spohn}&Parent (\citeyear{ddl:parent14,ddl:parent24})\\
transitivity+ \newnew{totality} & \ref{spohn} &  $\times$& \cite{ddl:parent14}\\
interval order & \ref{dr} & \ref{dr}  & \cite{ddl:parent24}
\end{tabular}
\end{center}
$\phantom{xxxxxxxxxxxxxxxx}$\caption{Some correspondences}\label{results}
\end{table}

\new{To improve readability, we have used certain shortcuts, albeit with the potential drawback of simplifying the data. The lack of correspondence in the 1st, 2nd and 5th row (starting from the top, going  downwards) is for the general case, when no constraint is put on $\succeq$. Thus, assuming one of reflexivity, totality or transitivity (under the max) does not add new validities. Similarly, the correspondence for limitedness is independent of any other properties (or axioms) in the background. The correspondence results for smoothness, transitivity (under the opt), transitivity+totality (under the max) and  interval order assume \ref{cod} and limitedness in the background. Quasi-transitivity, Suzumura consistency and acyclicity are known not to correspond to any formula in the general case, under the max rule$-$\cite{ddl:parent24}. This holds whether or not limitedness or smoothness is assumed in the background. However, it is not known what happens under the opt rule. Therefore, we have put these three conditions aside.

}

\begin{figure}
\includegraphics[scale=0.23]{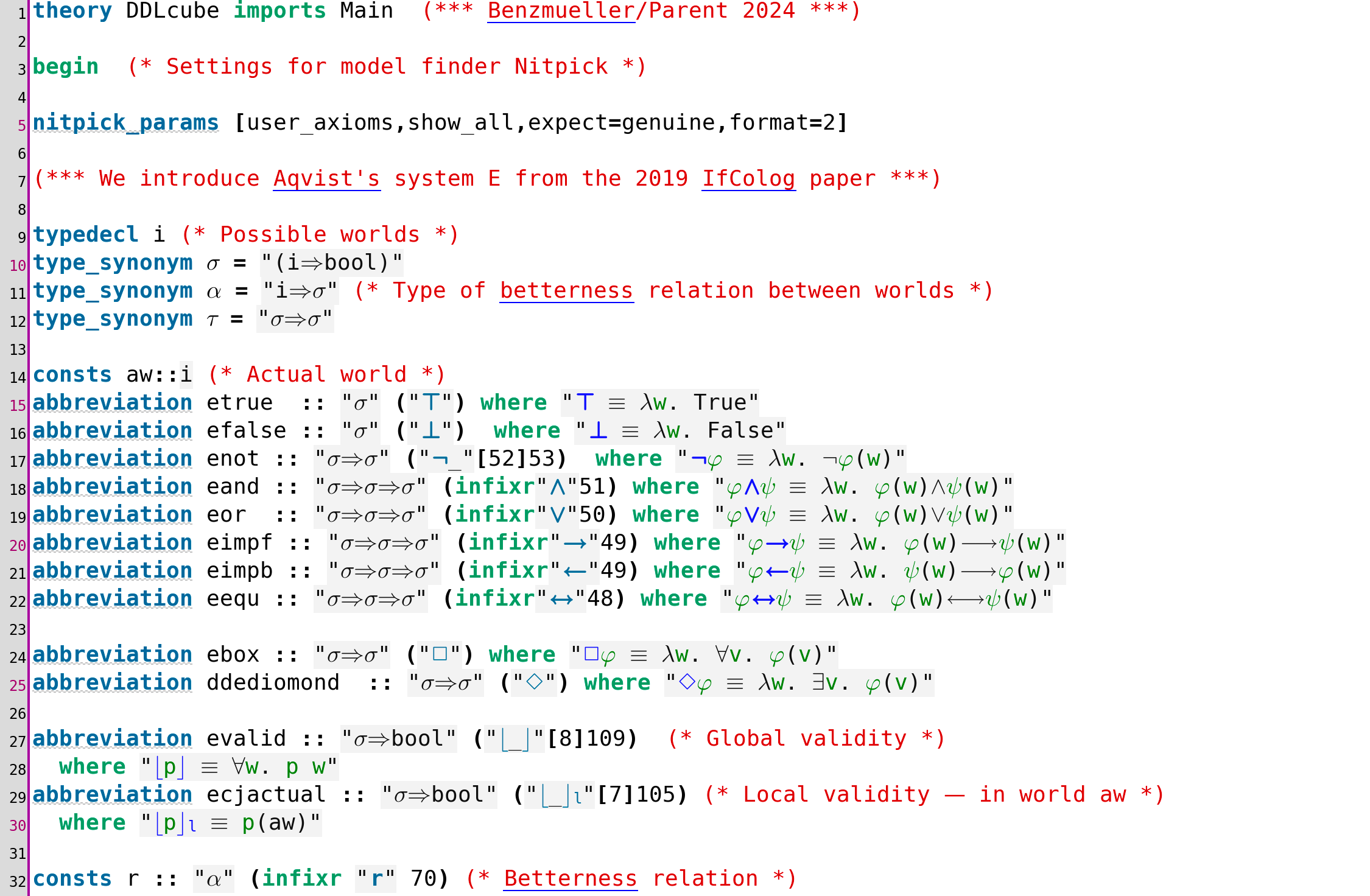}\caption{Basic semantical ingredients; propositional and modal connectives} \label{base}
\end{figure}
\begin{figure}[h]
\includegraphics[scale=0.23]{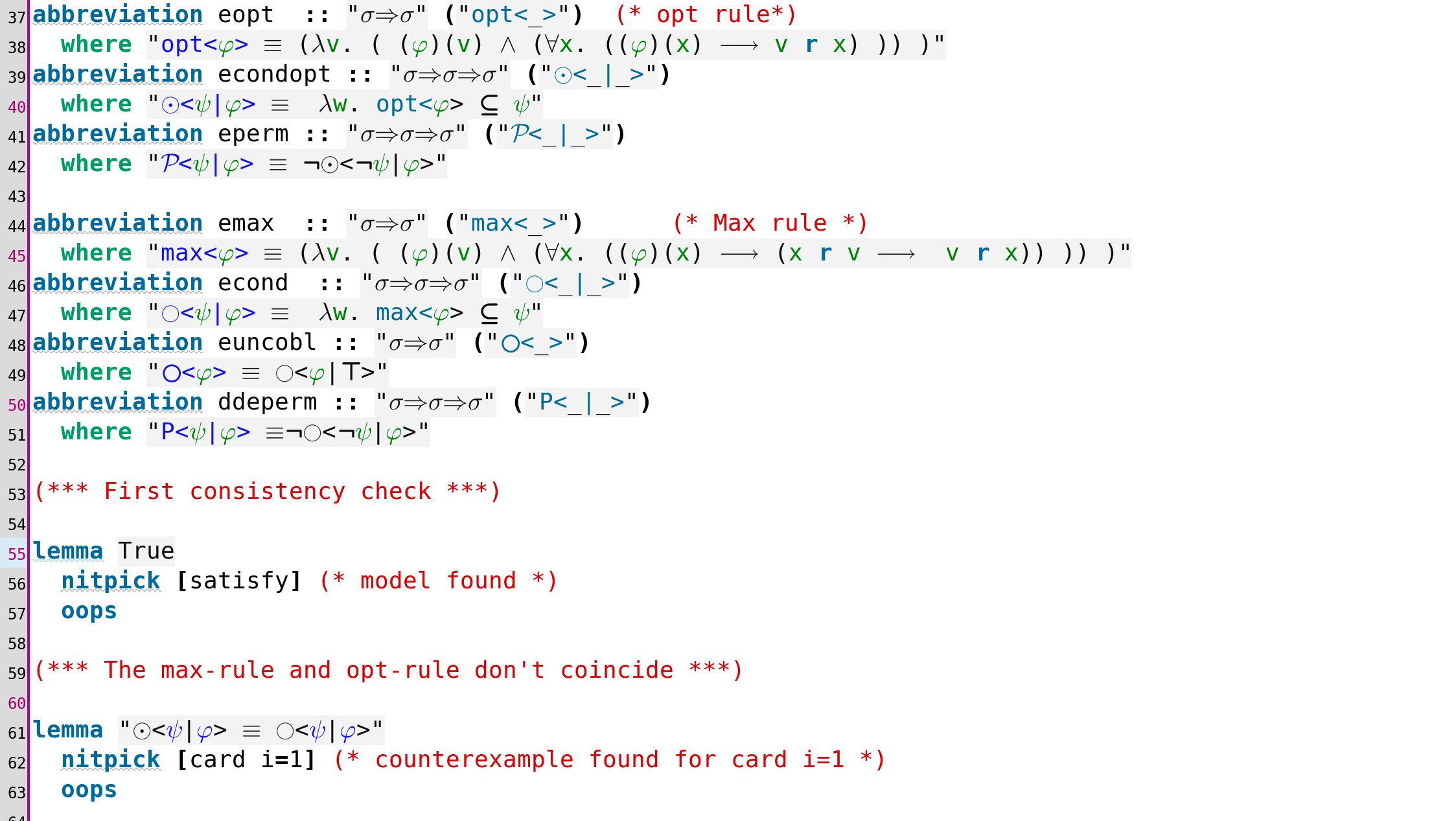}
\caption{Truth conditions}\label{rules}
\end{figure}

\section{System {\bf E} in Isabelle/HOL}
\label{hol}

Our modelling of System {\bf E} in Isabelle/HOL reuses and adapts prior work \citep{J45} and it instantiates and applies the LogiKEy methodology \citep{J48}, which 
supports plurality at different modelling layers.

\subsection{LogiKEy} \label{logikey}
Classical higher-order logic (HOL) is fixed in the LogiKEy methodology and infrastructure \citep{J48} as a \textit{universal meta-logic} \citep{J41} at the base layer (L0), on top of which a plurality of (combinations of) object logics can become encoded (layer L1). In the case of this paper, we encode extensions of System {\bf E} at layer L1 in order to assess them.
Employing these object logics notions of layer L1 we can then articulate a variety of logic-based domain-specific languages, theories and ontologies at the next layer (L2), thus enabling the modelling and automated assessment of different application scenarios (layer L3). Note that the assessment studies conducted in this paper at layer L3 do not require any further knowledge to be provided at layer L2; hence layer L2 modellings do not play a role in this paper.


LogiKEy significantly benefits from the availability of theorem provers for HOL, such as Isabelle/HOL, which internally provides powerful automated reasoning tools such as  \textit{Sledgehammer} (Blanchette et al., \citeyear{Sledgehammer}; Blanchette et al., \citeyear{blanchette2016hammering}) and \textit{Nitpick} \citep{Nitpick}. The automated theorem proving systems integrated via \textit{Sledgehammer} include higher-order ATP systems, first-order ATP systems, and SMT (satisfiability modulo theories) solvers, and many of these systems in turn use efficient SAT solver technology internally.
 Proof automation with \textit{Sledgehammer} and (counter-)model finding with \textit{Nitpick} were invaluable in supporting our exploratory modeling approach. 
These tools were very responsive in automatically proving (\textit{Sledgehammer}), disproving (\textit{Nitpick}), or showing consistency by providing a model (\textit{Nitpick}).
In this section and subsequent ones, we highlight some explicit use cases of \textit{Sledgehammer} and \textit{Nitpick}. They have been similarly applied at all levels as mentioned before.

\subsection{Faithful embedding of system {\bf E}}

In the work of \cite{J45}, it is shown that the embedding of {\bf E} in Isabelle/HOL is faithful, in the sense that a 
formula $\varphi$ in the language of ${\bf E}$ is valid in the class PREF of all preference models if and only if the HOL translation of $\varphi$  (notation: $\lfloor \varphi \rfloor$) is valid in the class of Henkin models of HOL. 

\begin{theorem}[Faithfulness of the embedding]\label{faith}
  \[\models^{\mathrm{PREF}} \varphi \text{ if and only if }
  \models^\text{HOL} 
  \lfloor \varphi \rfloor\]
\end{theorem}


Remember that the establishment of such a result is our main success criterium at layer L1 in the LogiKEy methodology.



This first two screenshots show the encoding of {\bf E} in Isabelle/HOL. Fig. \ref{base} shows the basic ingredients in the preference model, and describes how the propositional and alethic modal connectives are handled. The betterness relation $\succeq$ is encoded as a binary relational constant $r$ (l. 32).
In Fig. \ref{rules}, the notions of optimality and maximality are encoded. 
Different pairs of modal operators (obligation, permission) are introduced to distinguish between the two types of truth conditions.  
The model finder \emph{Nitpick} is able to verify the consistency of the formalization (l. 55) and to verify the non-equivalence between the two types of truth conditions (l. 61). 
\emph{Sledgehammer} is able to show the validity of all the axioms of {\bf E}. \newnew{This is shown in Fig.~\ref{e}  for the max rule. It takes only a few ms for some provers to prove a formula. For instance cvc4 shows \ref{abs} in 1ms, and \ref{sh} in 10 ms.}  
\begin{figure}[ht]
    \centering
   \includegraphics[scale=0.23]{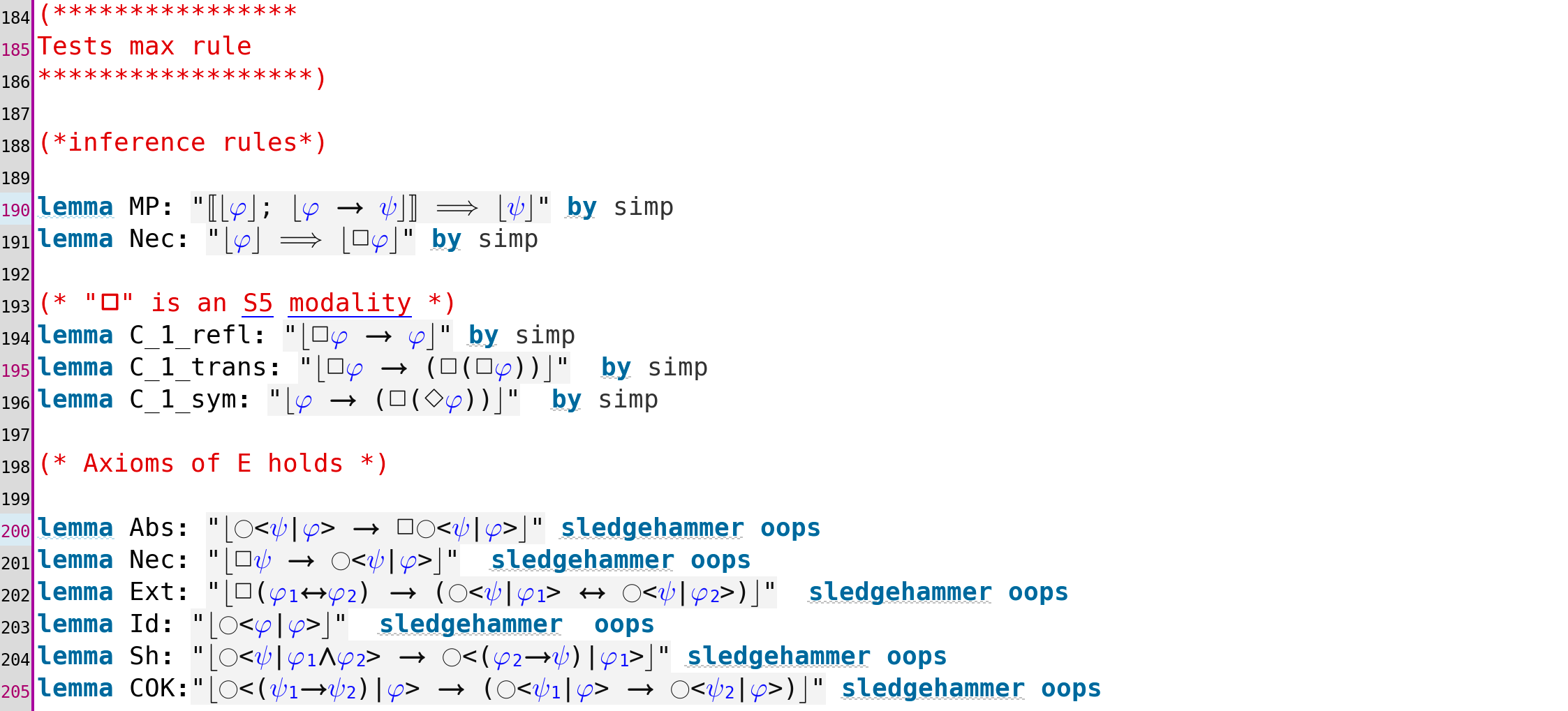}
   \caption{Axioms of {\bf E} (max)}
    \label{e}
\end{figure}

\subsection{Properties} 

The encoding of the properties of the betterness relation are shown in Figs. \ref{prop} and~\ref{trans}. On l. 99-104 of Fig. \ref{prop}, one sees the different versions of Lewis' limit assumption.
\begin{figure}[ht]
    \centering
   \includegraphics[scale=0.48]{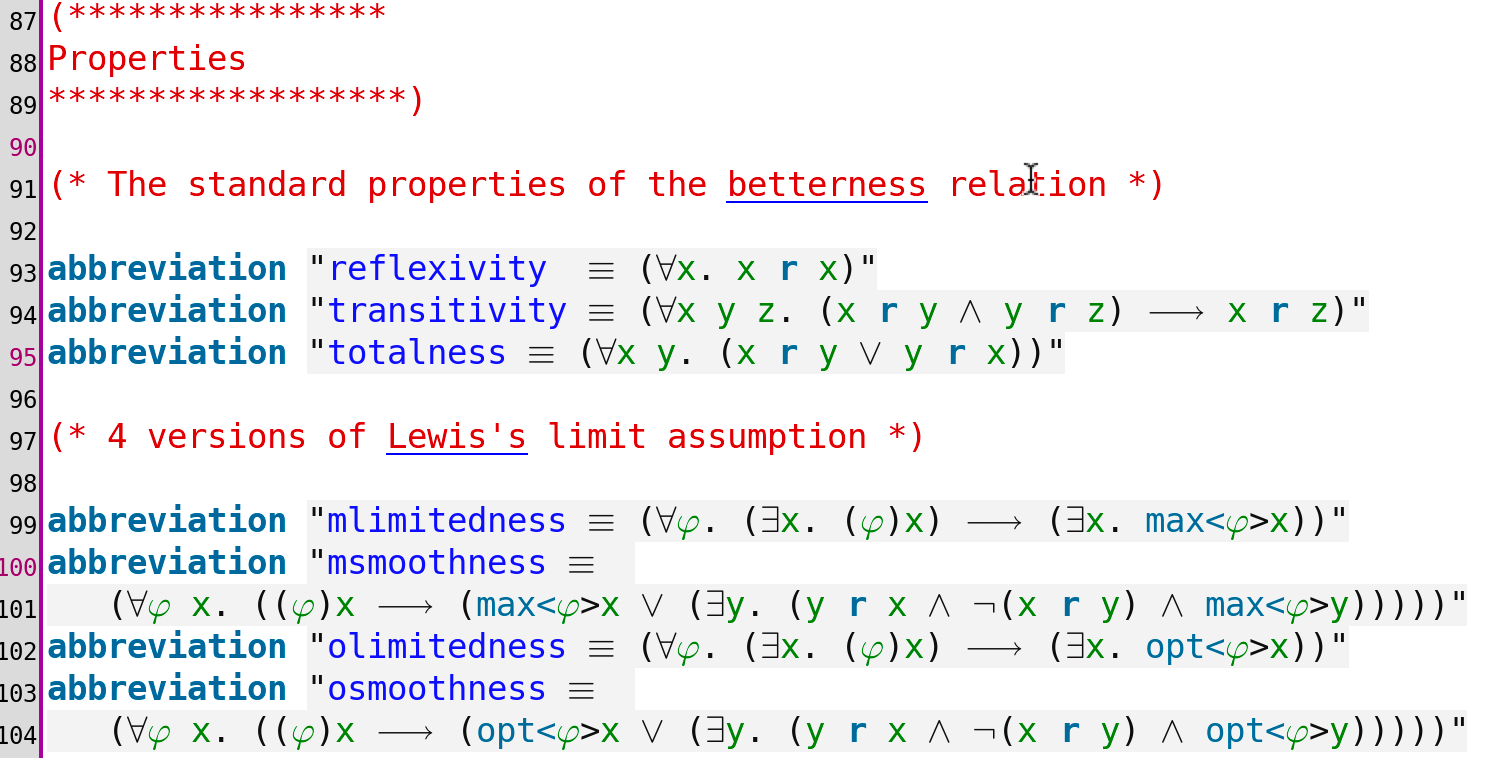}
   \caption{Standard properties}
    \label{prop}
\end{figure}
The property in Fig. \ref{trans} is the interval order condition. This one is usually described as the combination of totality with the Ferrers condition encoded in l. 146. 
 \emph{Sledgehammer}
confirms a fact often overlooked in the literature, that totality can be replaced by the simpler condition of reflexivity (l. 149-152). 
\begin{figure}[ht]
    \centering
   \includegraphics[scale=0.21]{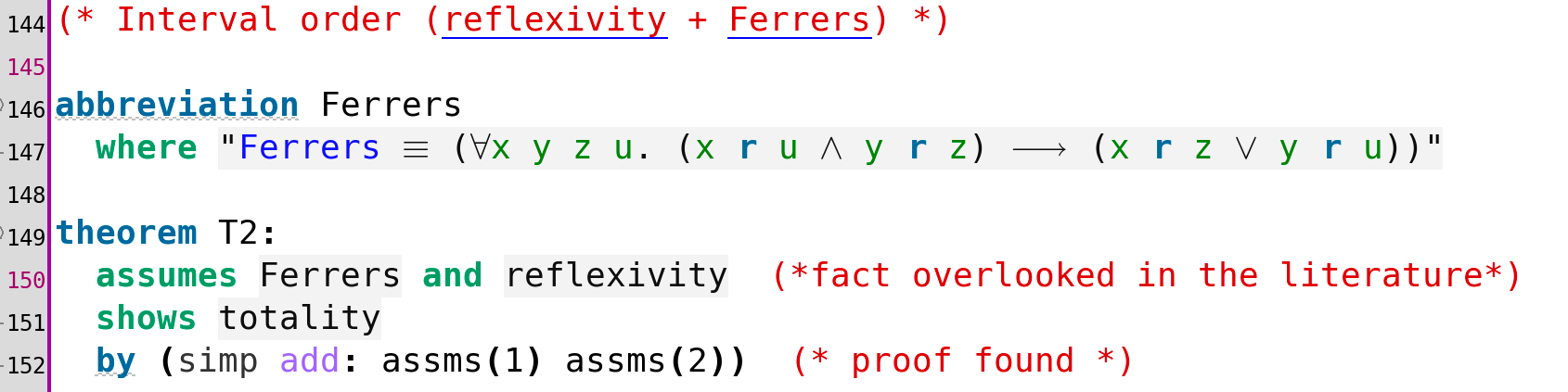}
   \caption{Interval order}
    \label{trans}
\end{figure}
The other candidate weakenings of transitivity discussed earlier are also encoded in the theory file. For simplicity's sake, we only give the example of quasi-transitivity and acyclicity.  The encoding of the second is shown in Fig. \ref{a-cycli}. In Isabelle/HOL, the transitive closure of a relation can be defined in a few lines, shown in Fig.~\ref{tr}. The encoding of quasi-transitivity is shown in Fig. \ref{quasi}.

\begin{figure}[ht]
    \centering
   \includegraphics[scale=0.47]{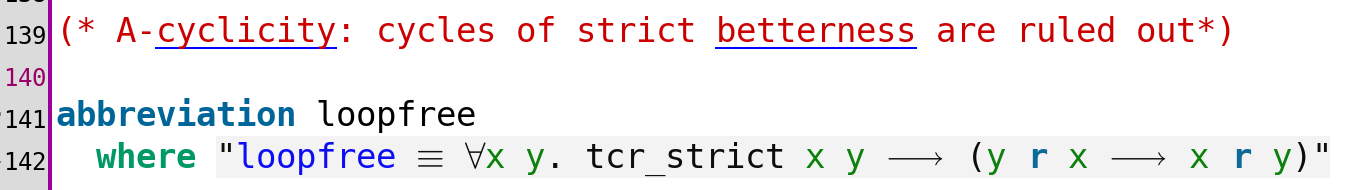}
   \caption{Acyclicity}
    \label{a-cycli}
\end{figure}

\begin{figure}[ht]
    \centering
   \includegraphics[scale=0.22]{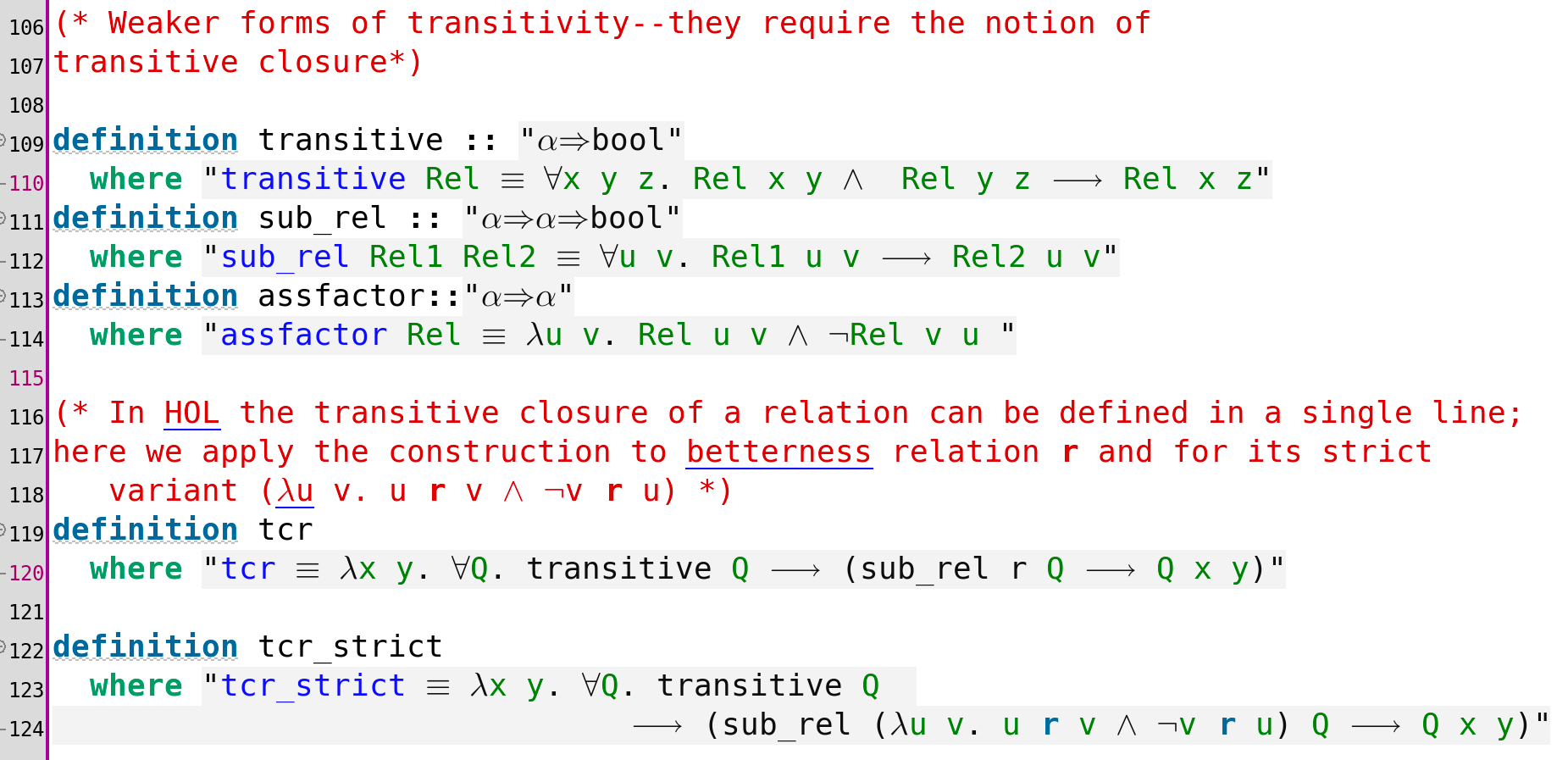}
   \caption{Transitive closure}
    \label{tr}
\end{figure}
\begin{figure}[ht]
    \centering
   \includegraphics[scale=0.5]{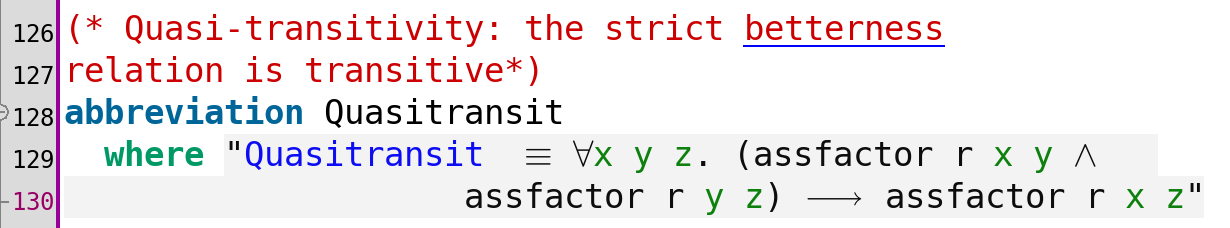}
   \caption{Quasi-transitivity}
    \label{quasi}
\end{figure}

\section{Verifying the correspondences} \label{cor}

\new{In this section, the correspondences for the axioms are investigated.
The task is to automatically verify that a given property is sufficient for the validation of the corresponding axiom as per  Table \ref{results}. We begin by assuming that the truth conditions for the obligation operator are given in terms of maximality, go on to consider the case where they are given in terms of optimality, and finally extend the scope of our inquiry  to a well-known variant evaluation rule for the conditional due to  \cite{ddl:L73}. The three
evaluation rules collapse only in the presence of all the properties of the betterness relation, including limitedness (which famously Lewis rejected). The consideration of Lewis's evaluation rule will also be needed for the case study in Section 5.
}
\subsection{Max rule}

Here we check known correspondences for  the max rule. \emph{Sledgehammer}  and \emph{Nitpick} confirm that an axiom is not valid unless the matching property is assumed: 
\begin{itemize}
\item If the relevant property is not assumed, counter-models for the corresponding axiom (D$^{\star}$, CM, DR  and Sp) are found by \emph{Nitpick}. This is Figs.~\ref{bb}, \ref{interval} and  \,\ref{totalplustr};
\item If the property is assumed, then the corresponding axiom is proved by \emph{Sledgehammer}. Fig.\,\ref{lim-smoot} shows it for limitedness and smoothness, Fig. \,\ref{interval} for the interval order condition, and Fig.\,\ref{totalplustr} for the combination of transitivity and totality.
\end{itemize}
The implications having the form ``property $\Rightarrow$ axiom'' are all verified. However, the converse implications are all falsified by \emph{Nitpick}. We will come back to this point later on. 

\begin{figure}[ht]
    \centering
   \includegraphics[scale=0.23]{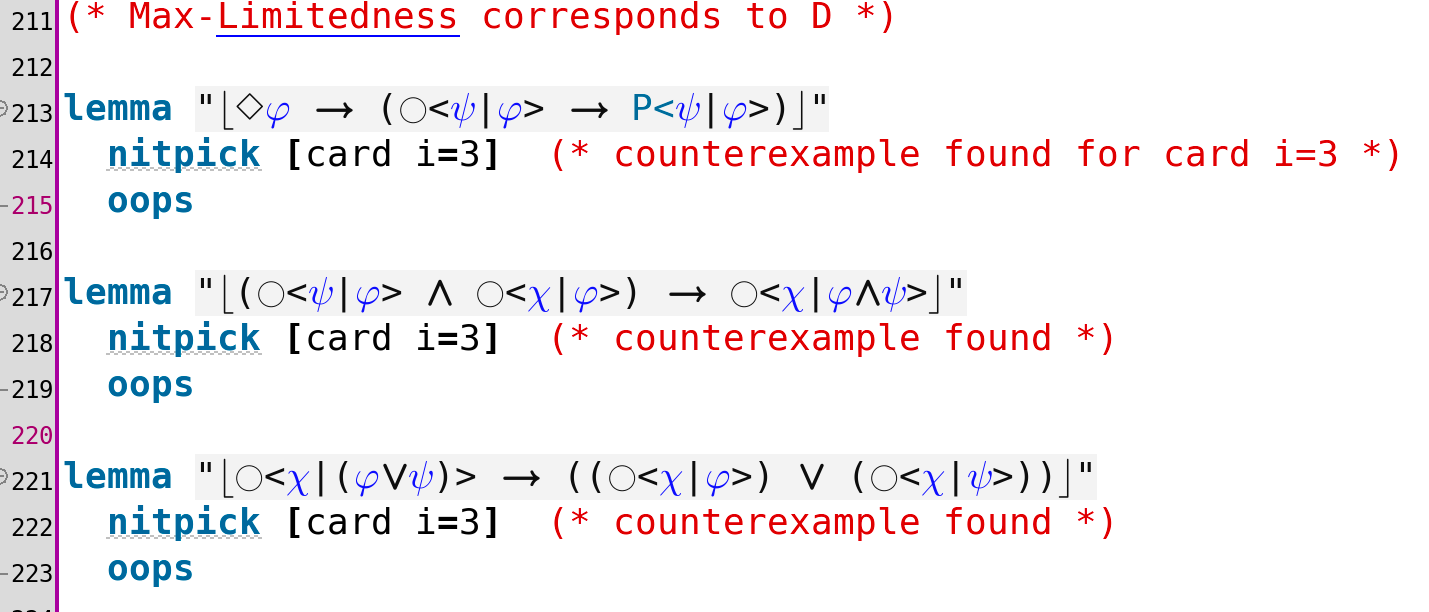}
  \caption{D$^{\star}$, CM and DR   invalid in general} \label{bb}
   \end{figure}

\begin{figure}[ht]

    \centering
   \includegraphics[scale=0.23]{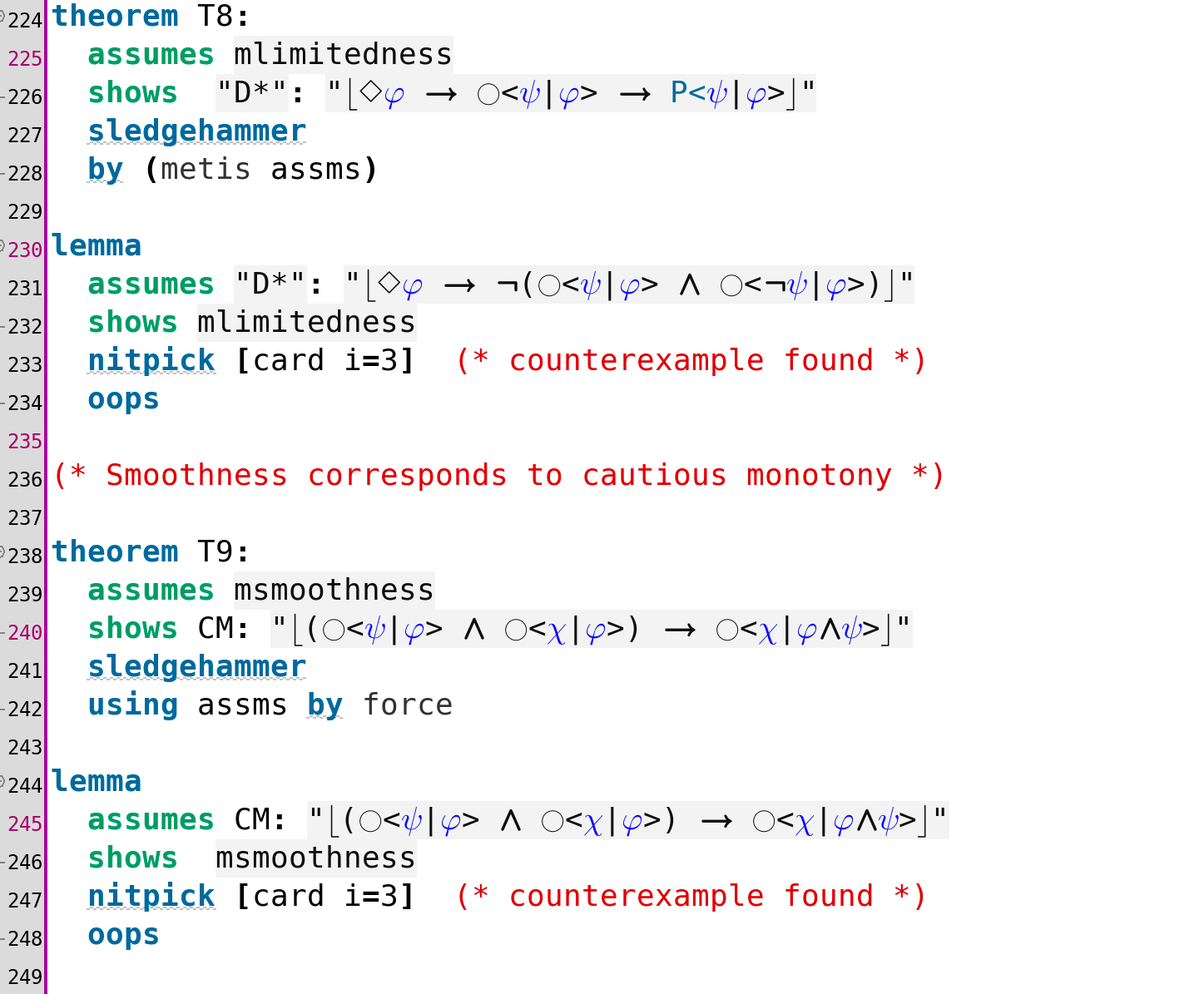}
   \caption{Limit assumption}
    \label{lim-smoot}
\end{figure}


\begin{figure}[ht]
    \centering
   \includegraphics[scale=0.22]{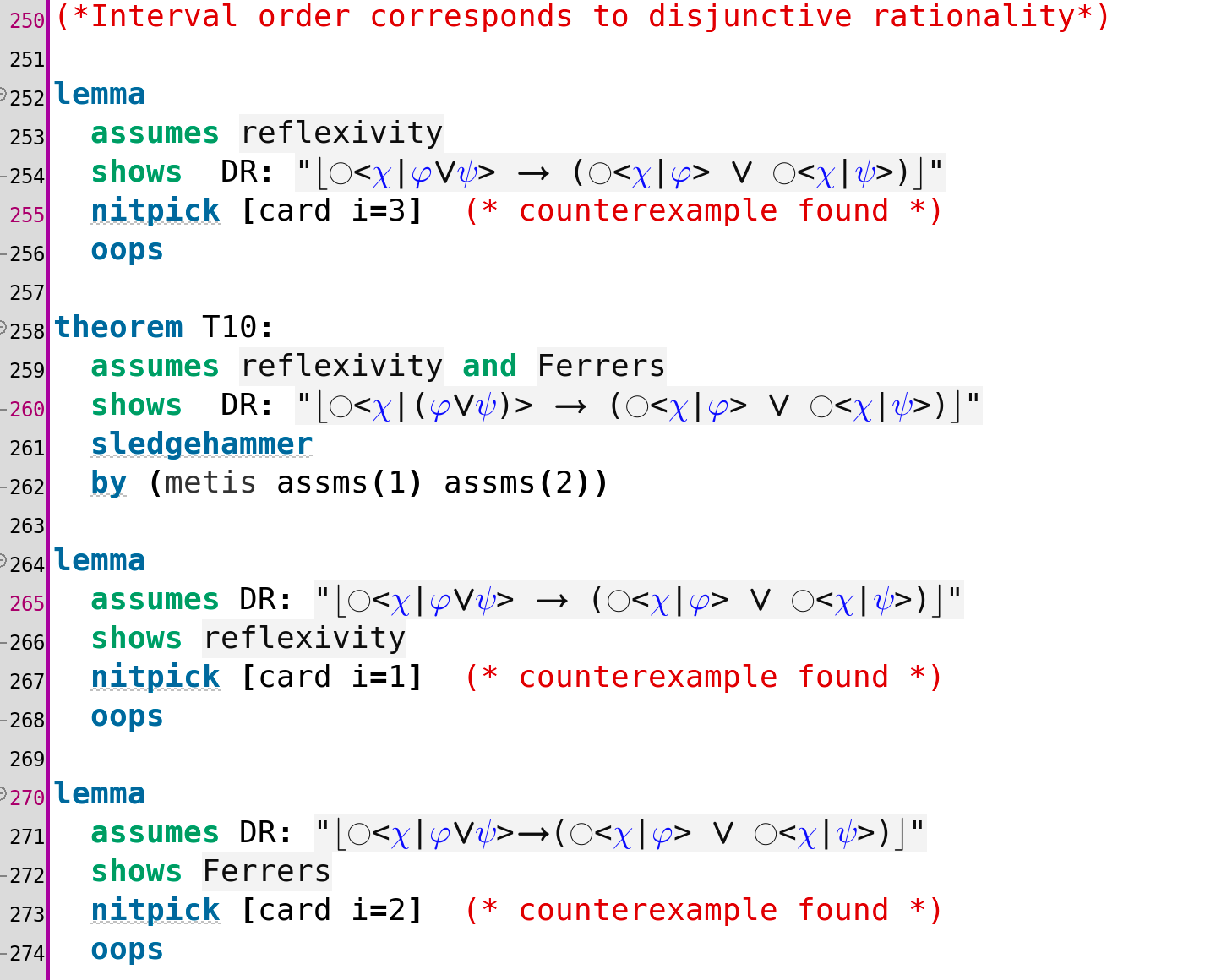}
   \caption{Interval order}
    \label{interval}
\end{figure}
\begin{figure}[ht]
    \centering
   \includegraphics[scale=0.22]{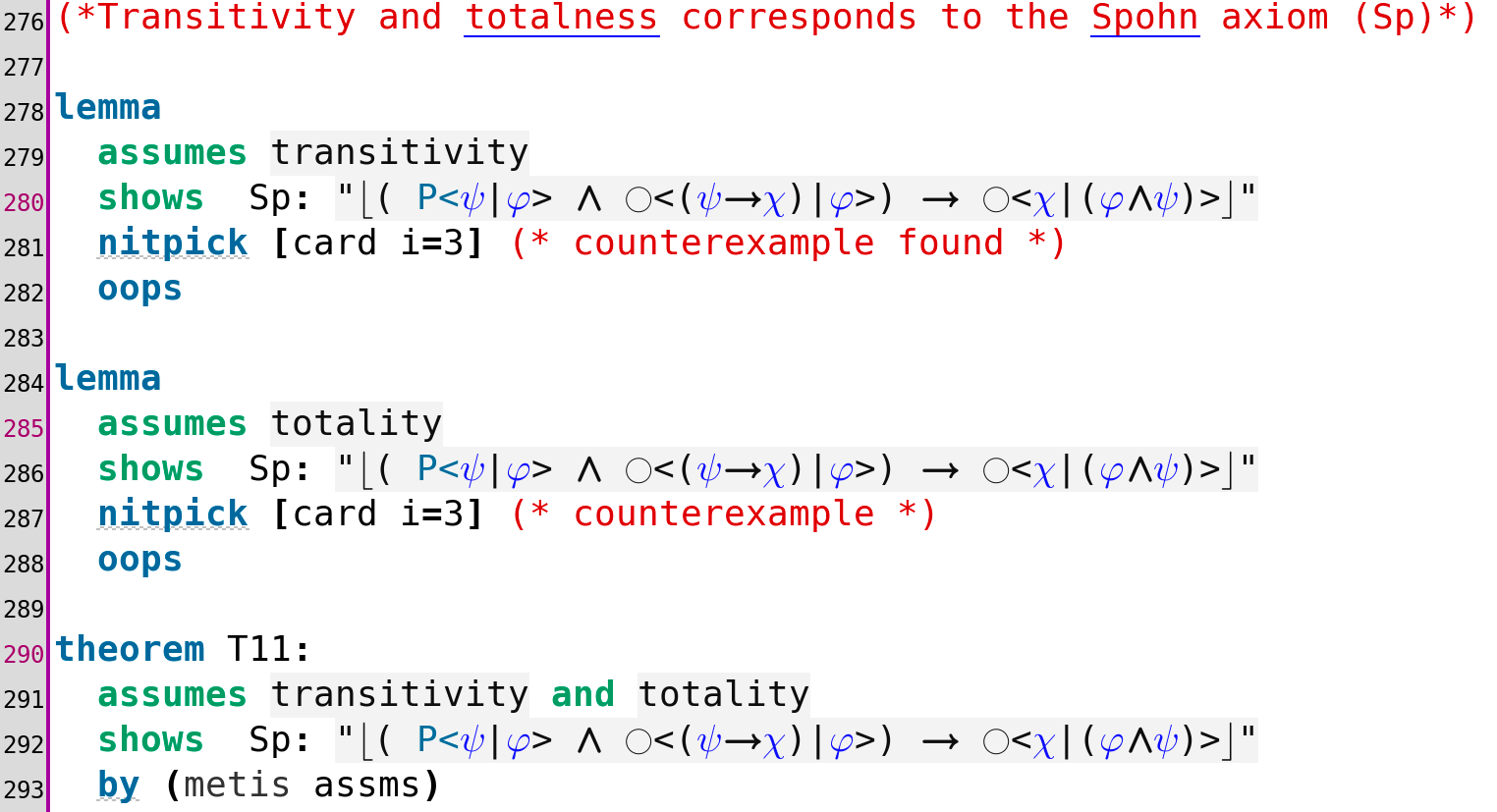}
   \caption{Transitivity and totality (max)}
    \label{totalplustr}
\end{figure}
\begin{figure}[ht]
    \centering
   \includegraphics[scale=0.22]{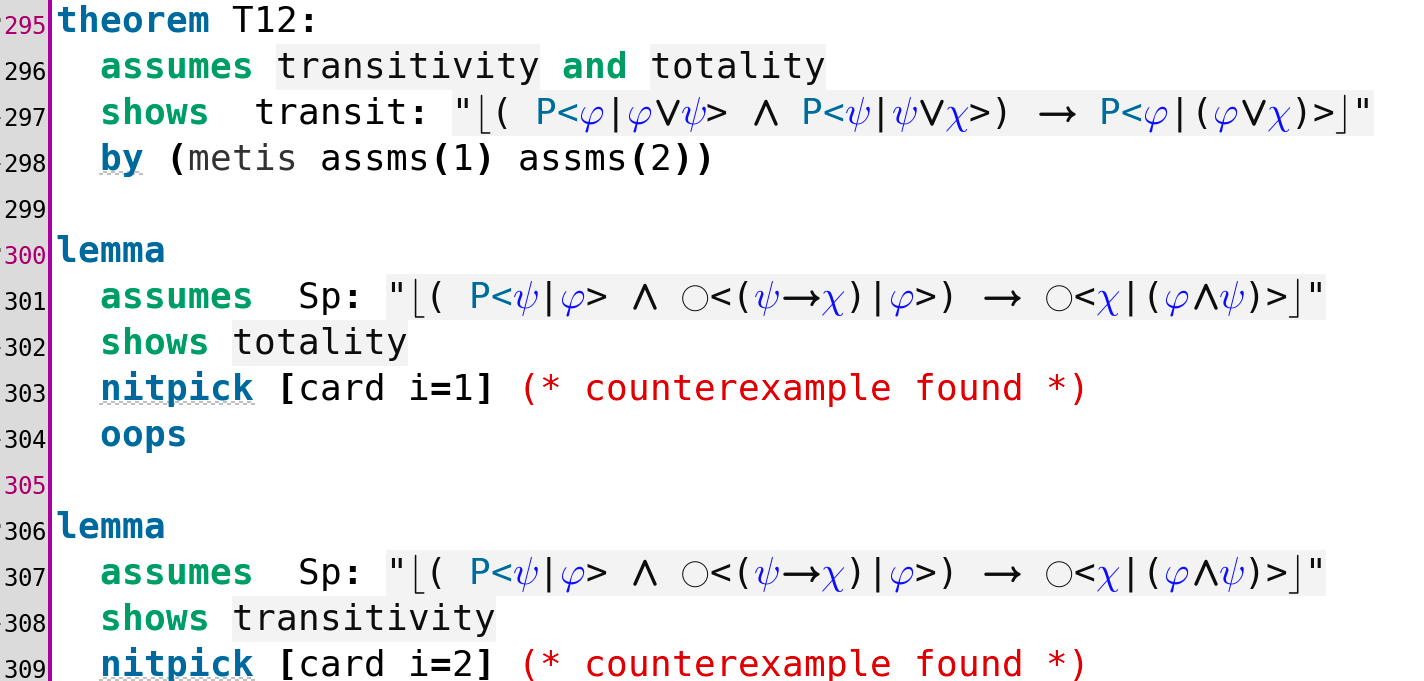}
   \caption{Transitivity and totality (max, cont.)}
    \label{totalplustr2}
\end{figure}




\subsection{Opt rule}

The outcome of our experiment is the same as for the max rule except for one small change. Transitivity no longer needs \new{totality} to validate Sp. This one only needs transitivity. Besides, the assumption of transitivity of the betterness relation gives us a principle of transitivity for a weak preference operator over formula\newnew{s}, defined by $\varphi\geq\psi$  iff $P(\varphi/\varphi\vee\psi)$.
This is shown in~Fig.~\ref{transi-opt}.  Again the converse implication is falsified by \emph{Nitpick}
(l. 356-361).
\begin{figure}[ht]
    \centering
   \includegraphics[scale=0.22]{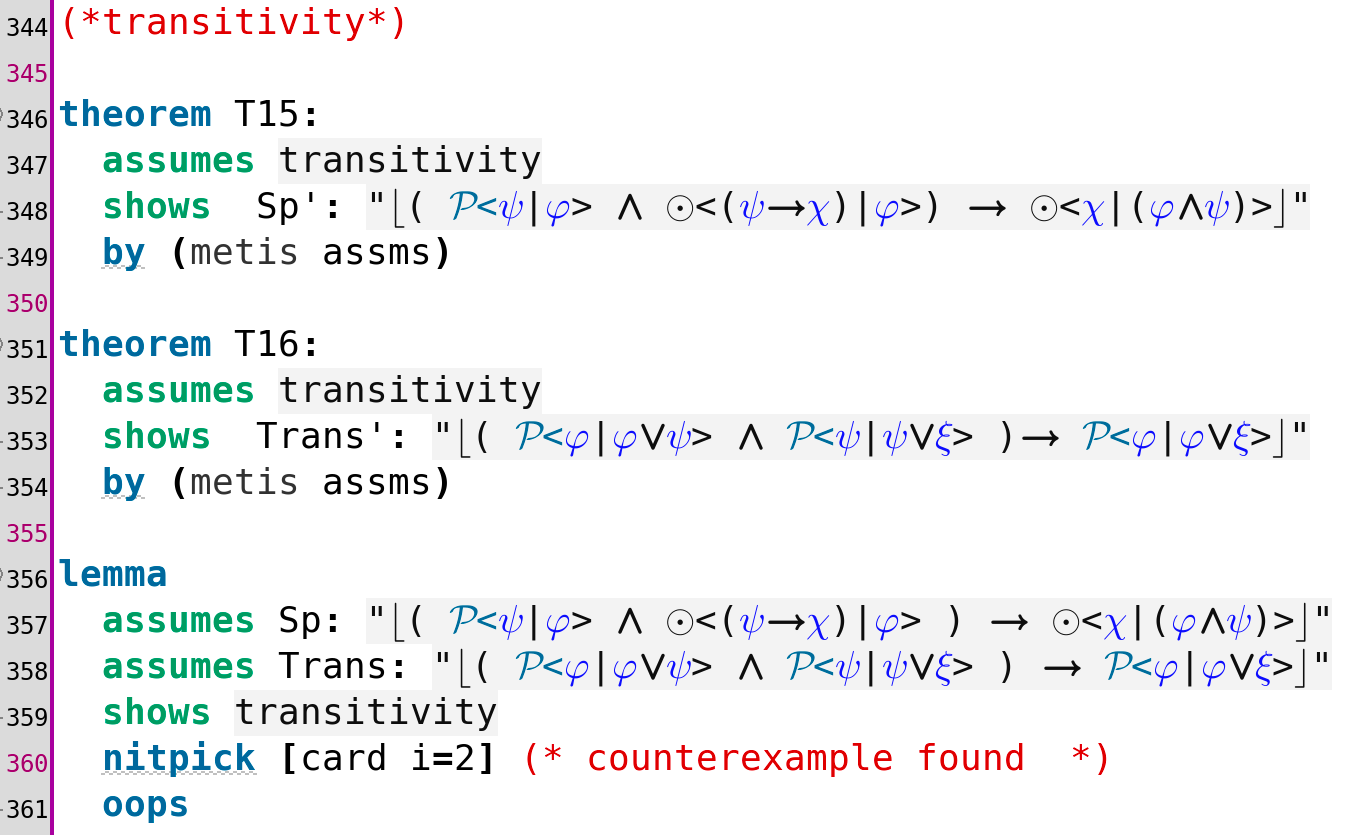}
   \caption{Transitivity (opt)}
    \label{transi-opt}
\end{figure}

\subsection{Inclusion}
In the work of \cite{C47},  proper inclusion between  systems in the modal cube are verified by looking at the model constraints of their respective axiomatizations. Because of the lack of full equivalence between modal axiom and property of  the relation, we cannot do the same, at least not yet.  Nor can we show equivalence between systems when restraining the number of worlds.

\newpage

\subsection{The \ref{lewis} truth conditions (Lewis)}\label{lew}

\new{We extend the scope of our inquiry to other truth conditions for the conditional}. Here we consider the variant rule proposed by \cite{ddl:L73}.
In order to avoid any commitment to the limit assumption, he suggests that $\bigcirc(\psi/\varphi)$ should be true whenever there is no $\varphi$-world or there is a $\varphi\wedge\psi$-world which starts a (possibly infinite) sequence of increasingly better \new{$\varphi\rightarrow\psi$-worlds, in which the obligation is never violated}. Formally:
\begin{flalign}\begin{split}
a\vDash \bigcirc (\psi/\varphi) \mbox{ iff } &   
\neg \exists b \;\;(b\models \varphi) \mbox{ or } \\ &
   \exists b \;\;(b\models\varphi \wedge \psi \;\;\& \;\;\forall c \;(c\succeq b \Rightarrow  c 
\models \varphi\rightarrow \psi)) \end{split}\tag{$\exists\forall$}\label{lewis}
\end{flalign}
We shall refer to the statement appearing at the right-hand-side of ``iff" as the \ref{lewis} rule.
The encoding is shown in Fig.\ref{lew-rule}.  
\begin{figure}[ht]
    \centering
   \includegraphics[scale=0.22]{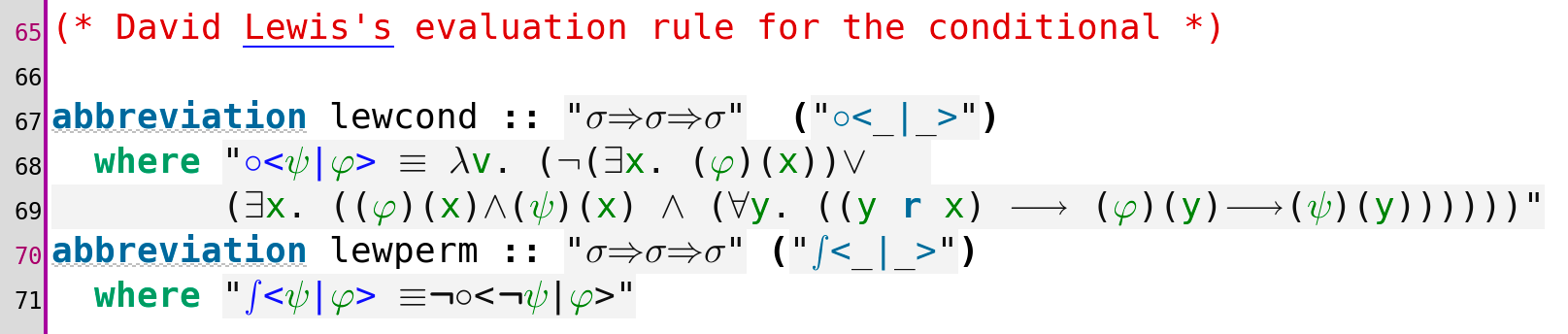}
   \caption{\ref{lewis} rule}
    \label{lew-rule}
\end{figure}

Isabelle/HOL can verify in what sense the standard account in terms of best requires the limit assumption. The law ``from $\Diamond\varphi$, $\bigcirc (\psi/\varphi)$ and $\bigcirc(\neg\psi/\varphi)$ infer $\bigcirc (\chi/\varphi)$" is valid.  This is known as the principle of ``deontic explosion'', often called DEX. It says that in the presence of a conflict of duties (unless it is triggered by an ``inconsistent" state of affairs) everything becomes obligatory. This has led most authors to make the limitedness assumption in order to validate D*, and hence make DEX harmless: the set $\{\Diamond\varphi, \bigcirc (\psi / \varphi),\bigcirc(\neg\psi/\varphi)\}$ is not satisfiable. This is shown in Fig.~\ref{dex}. On l. 394, the validity of DEX is established under the max rule. On l. 398, DEX is falsified under the \ref{lewis} rule.
\begin{figure}[ht]
    \centering
   \includegraphics[scale=0.22]{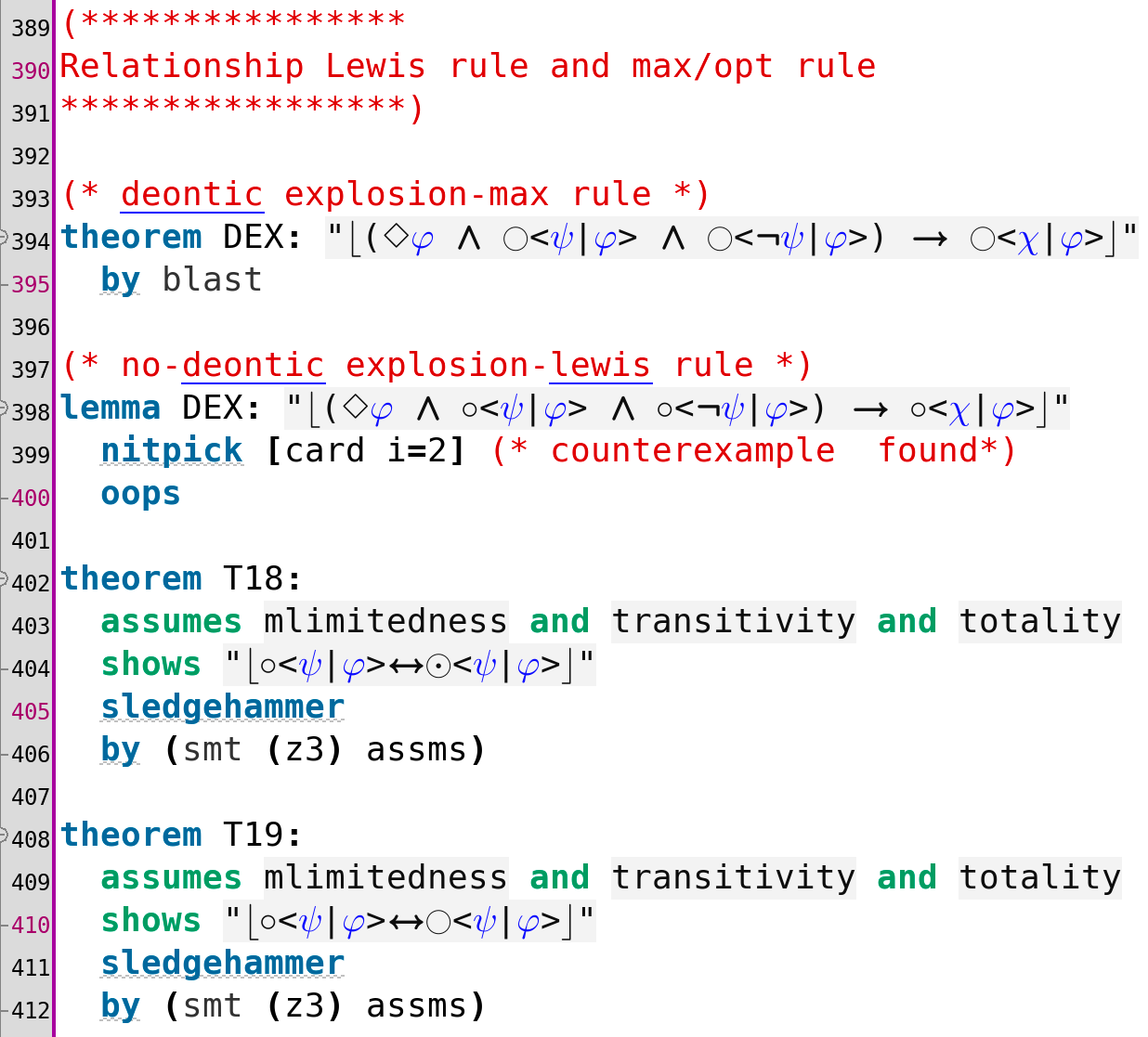}
   \caption{Deontic explosion (DEX)}
    \label{dex}
\end{figure}

Isabelle/HOL is also able to verify that when all the standard properties of the betterness relation are assumed, then the three evaluation rules collapse.  This is shown in Fig. \ref{dex} too.  T18 shows the equivalence between the \ref{lewis} rule and the opt rule, and T19 shows
the equivalence between the \ref{lewis} rule and the max rule.

Questions of correspondence between properties and modal axioms are still under investigation. There are two extra complications. First, a completeness result is available for the strongest system {\bf G} only: it is complete with respect to the class of models in which $\succeq$ is transitive and total (and hence reflexive).  Second, only two properties seem to have an import, but the matching between them and the axioms is not one-to-one: one property validates more than one axiom, sometimes in combination with the other property. This is shown in Table \ref{sound lewis}. The left column gives the axiom. The right column shows the property (or pair of properties) required to validate this one. 

\begin{figure}[h]
\centering
\includegraphics[scale=0.55]{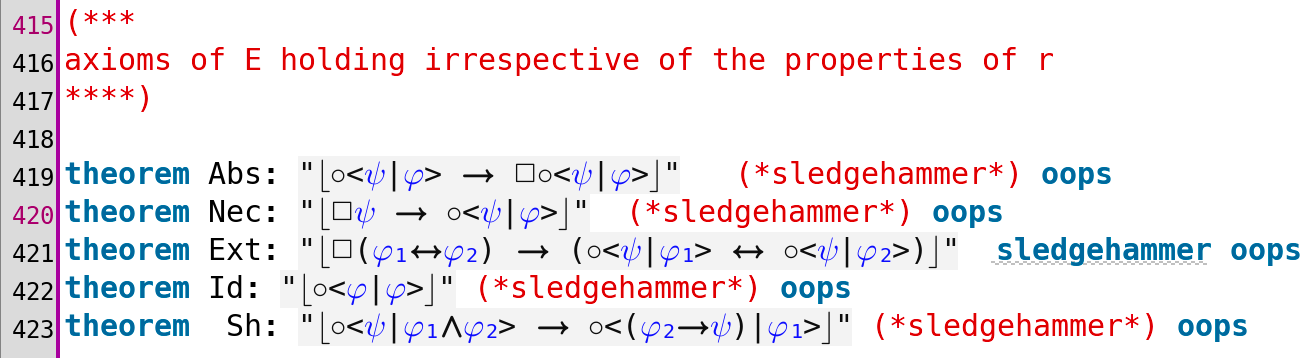}
\caption{Axioms independent of the properties (\ref{lewis} rule)}\label{e-lew}
\end{figure}
\begin{table}[ht]
	\centering
	\begin{tabular}{l|l}
		Axiom of {\bf G} & Property (or pair of properties)  of $\succeq$\\ \hline\hline
		(\ref{cod}) & totality\\
		(\ref{spohn}) & transitivity\\
		(\ref{cok}) & transitivity and totality\\
		(\ref{cm}) & transitivity and totality\\ \hline\hline
	\end{tabular} 
\caption{Axioms and properties under the \ref{lewis} rule$-$from \cite{ddl:P21}} \label{sound lewis}
\end{table}

\begin{figure}[h]
\centering
\includegraphics[scale=0.215]{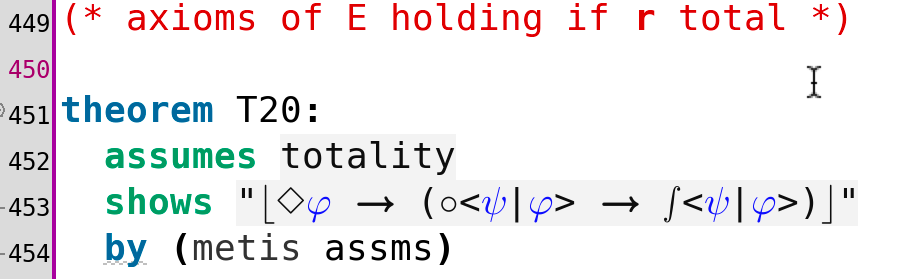}
\caption{Totality alone (\ref{lewis} rule)}\label{total-lew}
\end{figure}

\begin{figure}[h]
\centering
\includegraphics[scale=0.515]{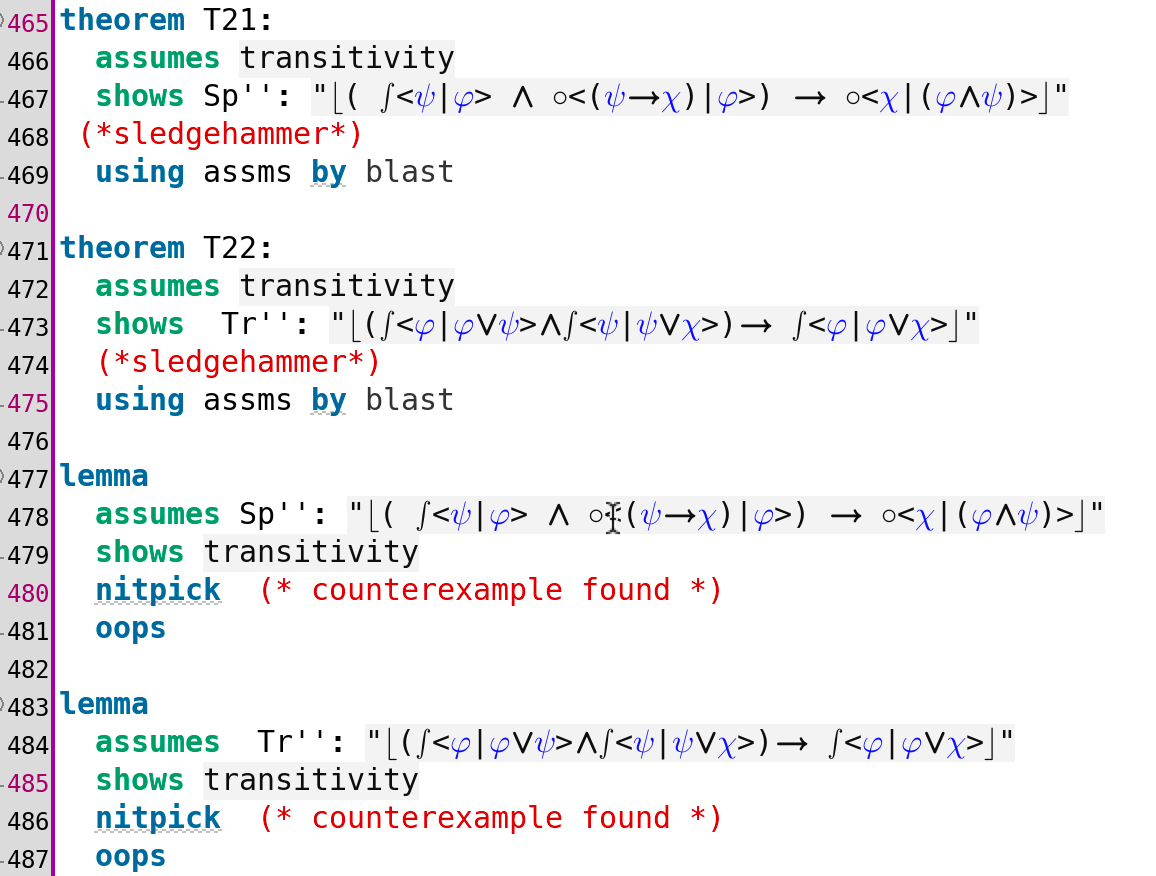}
\caption{Transitivity alone (\ref{lewis} rule, ct'd)}\label{total-lew2}
\end{figure}

\begin{figure}[h]
\centering
\includegraphics[scale=0.25]{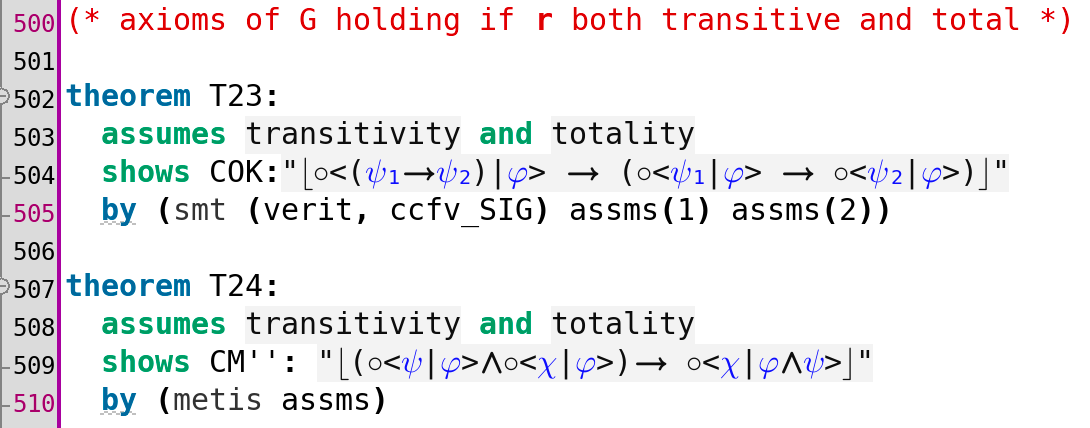}
\caption{Transitivity and totality together (\ref{lewis} rule)}\label{lew21}
\end{figure}


In Fig. \ref{e-lew},  \emph{Sledgehammer} shows the validity of the axioms of {\bf E} holding independently of the properties assumed of the betterness relation. In Figs. \ref{total-lew} and  \ref{total-lew2}, \emph{Sledgehammer} confirms that the \ref{cod} axiom and the \ref{spohn} axiom call for totality and transitivity, respectively. Similarly, Fig. \ref{lew21} shows that \ref{cok} and \ref{cm} call for \emph{both} transitivity and totality. In all these cases,  \emph{Sledgehammer} fails to show the converse implications.

\subsection{Assessment} \label{corres}

\begin{figure}[h]
\centering
\includegraphics[scale=0.5]{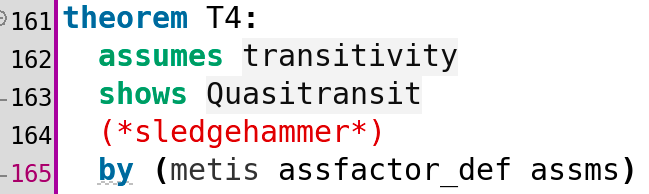}
\caption{Proving quasi-transitivity}\label{qxAI}
\end{figure}

\newnew{
As mentioned in Sect. \ref{logikey}, we found that the tools are very responsive in automatically proving (\textit{Sledgehammer}), disproving (\textit{Nitpick}), or showing consistency by providing a model (\textit{Nitpick}). 
Not only an answer is returned, but also a justification for this answer.  However, concerning this, we found that \textit{Nitpick} fares better than \textit{Sledgehammer}. 
When \textit{Sledgehammer} has proved a theorem successfully, the list of the definitions, axioms and lemmas to be used is returned. But the derivation itself is not given.\footnote{\newnew{We tried to use Isar style reconstruction of proofs with \emph{Sledgehammer}, but without much success for our examples.}}  In principle one could look into the detailed proof output file of the external provers called by \emph{Sledgehammer}, but this requires technical expertise. A simple example is given in Fig. \ref{qxAI}. Line 161 tells us that quasi-transitivity follows from transitivity by assumption (``assms") and using the definition of an asymmetric factor (``assfactor"). This is indeed how this would be shown by hand. However, a detailed argument is not given. 
By contrast, \textit{Nitpick} always gives the full details of the model justifying its answer, and this one was always correct in our experiments. An example of such a model is given in Fig. \ref{fig-cm}, which we will discuss in a moment. 

A comparative study with native provers, similar to the one in \cite{sb23}, must be left as a topic for future work. We are not aware of a similar automation of the systems studied in this paper using other methods.\footnote{\newnew{The hyper-sequent system for {\bf E} defined in \cite{DBLP:conf/prima/CiabattoniOP22} comes with a method for extracting a counter-model from a failed proof search. It holds the promise to provide a theorem prover against which we could evaluate the one described in this paper.}} A comparison with a prover for a related system (e.g. {KLML}ean 2.0, due to \cite{10.1007/978-3-540-73099-6_19}) would already be beneficial.


The entire Isabelle document (``DDLcube.thy") is verified by Isabelle2023 in 1m50s
on an Apple M1 with 8 GB of memory. During this time, Isabelle/HOL solves 82 problems, whereby demonstrating good responsiveness.  
It takes 15s for \emph{Nitpick} to find 40 (counter-)models, and 1m35s for 
\emph{Sledgehammer} to show the validity of  34 formulas and verify 6 implication relations between properties of the betterness relation. Additionally, we consistently observe accurate proofs and models, contrasting with the inherently error-prone nature of the traditional pen-and-paper method.
The assurance of accuracy is an added benefit of the known faithfulness of the embedding, 
Th.~\ref{faith}, a distinctive hallmark of the method. 




}

We end with a critical assessment of the findings on correspondence. The situation for conditional (deontic) logic is still slightly different from the one for traditional modal logic. In the latter setting, the full equivalence between the property of the relation and the modal formula is verified by automated means. In the former setting only the direction ``property $\Rightarrow$ axiom" is verified by automated means. To be more precise, what is verified is the fact that, if the property holds, then the axiom holds. What is not confirmed is the converse statement,  that if the axiom holds then the property holds.
This asymmetry deserves to be discussed. 

First, it is usual to distinguish between validity on a frame and validity in a model based on a frame. A frame is a pair $\mathcal{F}=(W,R)$, with $W$ a set of worlds and $R$ the accessibility relation. A model based on $\mathcal{F}=(W,R)$
is the triplet  $\mathcal{M}=(W,R,v)$ obtained by adding a specific valuation $v$, or a specific assignment of truth-values to atomic formulas at worlds. For a formula to be valid on a frame  $\mathcal{F}$, it must be valid in all models  based on $\mathcal{F}$. In other words, it must be true for every assignment to the atomic formulas.  We have worked at the level of models. But in so-called correspondence theory \emph{à la} Salqhvist-van Benthem, the link between formulas and properties is in general studied at the level of frames themselves. One shows that $\mathcal{F}$ meets a given condition iff formula $A$  is valid on  $\mathcal{F}$. In a recent extension of the semantical embedding approach for public announcement logic PAL, cf. \cite{J58}, an explicit dependency on the concrete evaluation domain has been modeled. 
The question as
to whether this idea can be extended to support a notion of frame-validity  is a topic for future research.




Second, the most we got is that a given property  is a sufficient condition for the validity of the axiom, but not a necessary one. For instance, to disprove the implication ``\ref{cm} $\Rightarrow$ m-smoothness" under the max rule (Fig.~\ref{lim-smoot}), \emph{Nitpick} exhibits a model in which \ref{cm} holds and m-smoothness \new{is} falsified. This model is shown in Fig.~\ref{fig-cm}. The corresponding preferential model is also shown below. Smoothness is falsified, because it contains an infinite loop of strict betterness, making the smoothness condition fail for, e.g., $\varphi\vee\neg\varphi$. But \ref{cm} (vacuously) holds, because the two conjuncts appearing in the antecedent of the axiom are both false. Indeed, $i_3$ is a maximal $\varphi$-world, and it falsifies $\psi$ and $\chi$.  This shows that m-smoothness is not a  necessary condition for the axiom to hold.

\begin{figure}
\centering
\begin{subfigure}{0.5\textwidth}
\centering
\includegraphics[width = \textwidth]{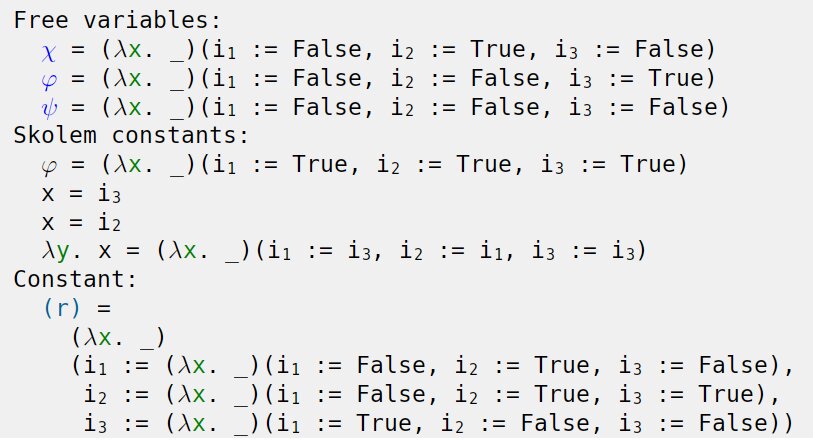}
\caption{Model for HOL}
\label{fig:left}
\end{subfigure}\vspace{1cm}
\begin{subfigure}{0.49\textwidth}
\centering
\includegraphics[scale=0.35]{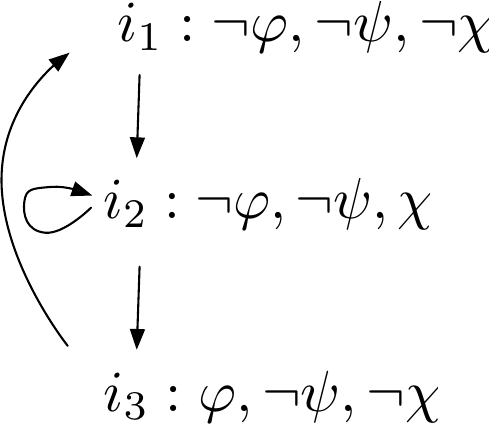}
\caption{Preferential model.  An arrow from $i_1$ to $i_2$ means  $i_1\succeq i_2$.  No arrow
from  $i_2$ to $i_1$ means  $i_2\not\succeq i_1$}
\label{fig:right}
\end{subfigure}
\caption{A non-smooth model validating \ref{cm} (max)}\label{fig-cm}
\end{figure}
It is interesting to remark that \emph{Nitpick} always presents a finite standard model. We leave it as a topic for future research to investigate if the crucial distinction between standard and non-standard models for HOL which, according to \cite{And02}, sheds so much light on the mysteries associated with the incompleteness theorems,  has a bearing on the issue at hand.


Another open problem concerns the possibility of verifying ``negative" results. 
As shown in Table \ref{results}, under the max rule transitivity alone does not correspond to any axiom. Also under both the max rule and the opt rule neither reflexivity nor totality correspond to an axiom. Finally, under the \ref{lewis} rule the limit assumption has no \new{impact}. All this has been established with pen and paper. It would be worth exploring the question as to whether and how this problem could be tackled in Isabelle/HOL.

\section{Case study: Parfit's repugnant conclusion} \label{rc}

In this section,\footnote{
\newnew{The code snippets are taken from the companion theory files ``mere\_addition\_opt.thy", ``mere\_addition\_max.thy'' and ``mere\_addition\_lewis.thy''.}} 
\newnew{we show how to employ the framework described in the previous sections} 
for the computer-aided assessment of ethical arguments in philosophy. 
Our focus is on analyzing the repugnant conclusion as discussed by  \cite{Parfit1984-PARRAP}.
We provide a computer encoding of his
argument for the repugnant conclusion to make it amenable to formal analysis and computer-assisted experiments.  \newnew{Through the use of Isabelle/HOL, we discuss the plausibility of a solution of the paradox, advocated by \cite{ddl:T87} and others. It involves rejecting the assumption of transitivity of ``better than''.
To put the proposed solution to the test, a full-blooded logical characterization of ``better than" is needed. This one is given by the framework described in the previous sections. Following the tradition in deontic and conditional logic (refer to, for example, \cite{ddl:L73}), we make a distinction between ``better than" as a relation on formulas and as a relation on possible worlds, with the latter being instrumental in defining the logic of the former.}

This section is organized into three subsections. Subsection \ref{par} describes Parfit's argument for the repugnant conclusion. For readability's sake, we focus on a simplified version of the paradox, called the mere addition paradox. Subsection \ref{exp} documents the experiments we have run. Subsection \ref{sum:f} summarizes our findings.

\subsection{The paradox} \label{par}


The repugnant conclusion reads:
	\begin{quote}		
		 ``For any possible population of at least ten billion people, all with a very high quality of life, there must be some much larger imaginable population whose existence, if other things are equal, would be better even though its members have lives that are barely worth living.'' \cite[Ch. 6]{Parfit1984-PARRAP}
			\end{quote}
The target is ``total utilitarianism", according to which 
the best outcome is given by the total of well-being in it.
This view implies that any loss in the quality of lives in a population can be compensated for by a sufficient gain in the quantity of a population. 
Fig.~\ref{rep:conc}  illustrates the repugnant conclusion. 
The blocks correspond to two populations, $A$ and $Z$. The width of each block represents the number of people in the corresponding population, the height represents their quality of life. All the lives in the above diagram have lives worth living. People’s quality of life is much lower in $ Z$ than in $A$ but, since there are many more people in $Z$, there is a greater quantity of welfare in $Z$ as compared to $A$. Consequently, although the people in $A$ lead very good lives and the people in $Z$ have lives only barely worth living, $Z$ is nevertheless better than $A$ according to classical utilitarianism. 
 \bigskip
\begin{center}
 \begin{figure}[h] %
\setlength{\unitlength}{0.3in}
\begin{picture}(0,3)
	\put(2,0){\framebox(0.5,3){$A$}}
	\put(0.5,3.2){{\small Very high quality in life}}
	
		\put(4,0){\framebox(5,0.3){$Z$}}
	\put(8,1){{\small Very low but positive quality in life}}
		\put(9.5,0.5){{\small $Z$ has a lot more people}}
\end{picture}
\caption{Repugnant conclusion}\label{rep:conc}
\end{figure}
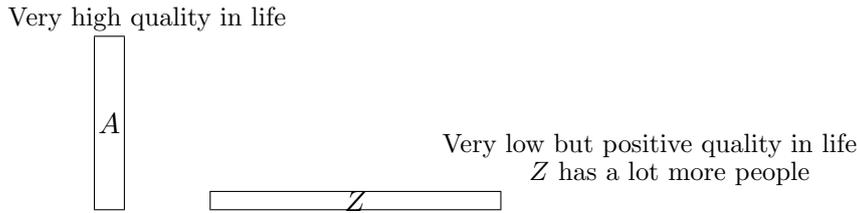	
\end{center}
It has been argued by e.g.  \cite{ddl:T87} 
that the repugnant conclusion can be blocked, by just dropping the assumption of the transitivity of ``better than''. This is best explained by considering a smaller version of the paradox, called the mere addition paradox. The repugnant conclusion is generated by iteration of the reasoning underlying the mere addition paradox. 

The mere addition paradox is shown in Fig. \ref{mere:add}. In population $A$, everybody enjoys a very high quality of life. In population $A^{+}$ there is one group of people as large as the group in $A$  and with the same high quality of life. But  $A^{+}$ also contains a number of people with a somewhat lower quality of life. In Parfit’s terminology $A^{+}$ is generated from $A$ by ``mere addition". Population $B$ has the same number of people as $A^{+}$, their lives are worth living and at an average welfare level slightly above the average in  $A^{+}$, but lower than the average in $A$. 
The link with the repugnant conclusion is that by reiterating this structure (scenario $B^{+}$ and $C$, $C^{+}$ etc.), we end up with a population $Z$ in which all lives have a very low positive welfare.
\begin{center}
\begin{figure}[h]\setlength{\unitlength}{0.3in}  \centering
\begin{picture}(8,4)
	\put(2,3.5){$A$}
		\put(4.5,3.5){$A^{^+}$}
		\put(7,3.5){$B$}
	
	\put(2,0){\framebox(0.5,3){$\;$}}
	\put(4,0){\framebox(0.5,3){$\;$}}
	\put(4.7,0){\framebox(0.5,1.5){$\;$}}

	 \linethickness{0.09mm}
  \multiput(6.7,2.25)(0.4,0){5}{\line(1,0){0.3}} 
  
\thinlines
		\put(4.5,-0.35){{\small + people}}
		\put(7,0){\framebox(1,2.5){$\:$}}
	\put(8.5,2.5){{\small average in $A^{+}$ }}
\end{picture}
\caption{Mere addition paradox}\label{mere:add}
\end{figure}
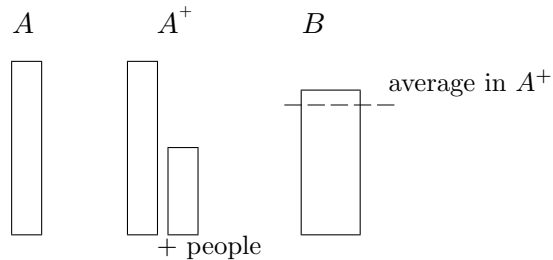	
\end{center}

The following statements are all plausible: 
\begin{itemize}
    \item [(P0)] $A$ is strictly better than $B$: $A > B$. Otherwise, in the original scenario, by parity of reasoning or consistency (scenario $B^{+}$ and $C$, $C^{+}$, ...) one would have to deny that $A$ is better than $Z$.
   \item [(P1)]     $A^{+}$ is at least as good as $A$: 
     $A^{+}  \geq  A$. Justification: $A^{+}$ is not worse than (and hence  at least as good as) $A$; the addition of lives worth living (the + people) cannot make a population worse. 
           \item [(P2)]
     $B$ is strictly better than $A^{+}$: 
     $B>A^{+}$. Justification: $A^{+}$ and $B$ have the same size; the average welfare level in $B$ is slightly above the average in $A^{+}$, and the distribution is uniform across members. So 
      $B$ is better in regard to both average welfare (and thus also total welfare) and equality. \end{itemize}


The relations $\geq$ and $>$  appearing in (P0)-(P2) apply to propositional formulas. It is usual to take $\varphi\geq\psi$ as a shorthand of  $P(\varphi/\varphi\vee\psi)$, and $\varphi>\psi$ as a shorthand of $\varphi \geq \psi$ and  $\psi\not\geq\varphi$. (Cf. \cite{ddl:L73}). This is shown in Table \ref{pref:form}.


\medskip
\begin{center}
\begin{table}[h]
\centering
\begin{tabular}{c|c|c}
Definiendum     & Definiens & Reading  \\ \hline
$\varphi \geq \psi$     &  $ P (\varphi/\varphi\vee\psi) $ & $\varphi$  permitted, if $\varphi\vee \psi$ \\
$\varphi > \psi$     &  $ P (\varphi/\varphi\vee\psi) \wedge$ $\bigcirc (\neg\psi/\varphi\vee\psi)$ & $\varphi$ permitted and $\psi$ forbidden, if $\varphi\vee \psi$\\
\end{tabular}
\vspace{0.2cm}
\caption{Preference on formulas}\label{pref:form}
\end{table}
\end{center}



\vspace{-1cm}
Fig. \ref{fig:mereadd} shows the encoding of P0-P2 in terms of obligation, an obligation statement being evaluated using the opt rule: 
\begin{figure}[ht]
    \centering  
    \includegraphics[scale=0.5]{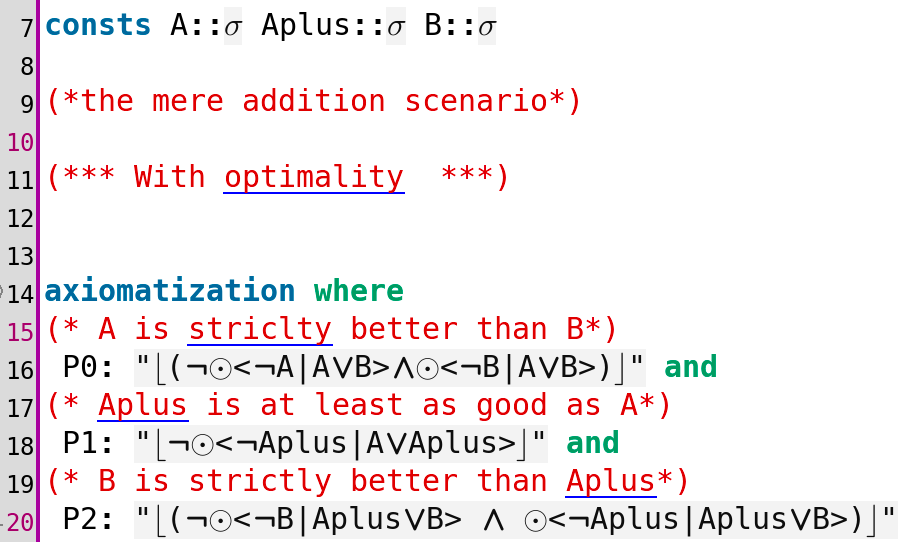}
    \caption{Encoding of the mere addition scenario (optimality)}
    \label{fig:mereadd}
\end{figure}

\begin{figure}[h]
    \centering
    \includegraphics[scale=0.5]{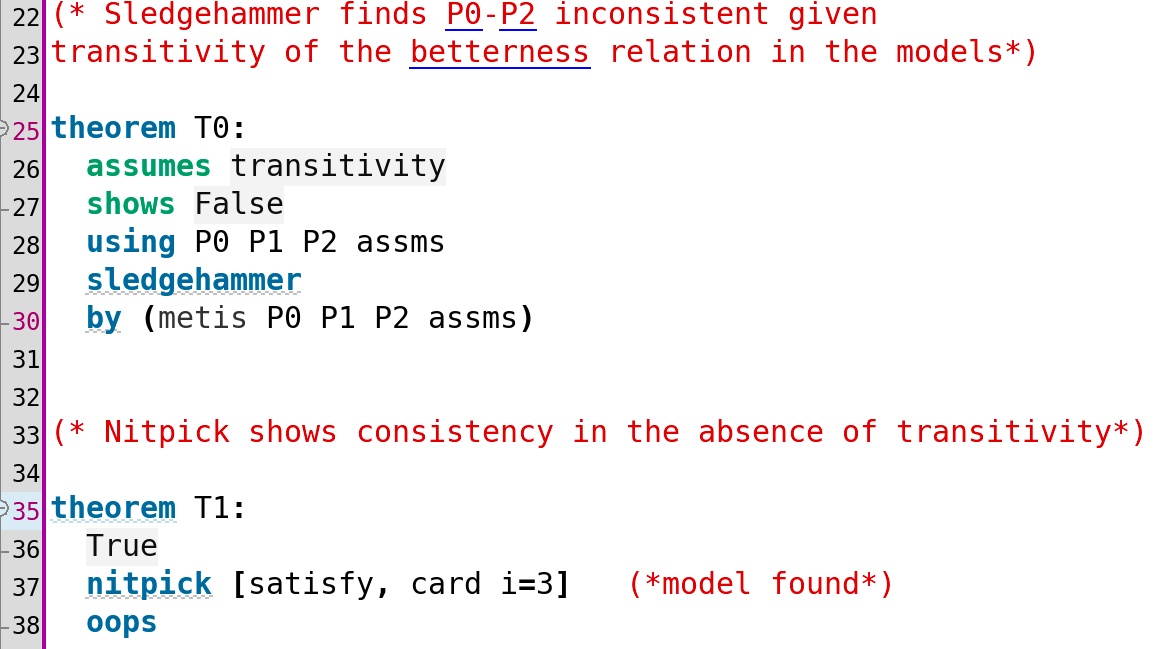}
    \caption{Sample queries on (P0)-(P2)}
    \label{fig:mereadd:queries}
\end{figure}

\subsection{Computer-assisted experiments}\label{exp}

Fig. \ref{fig:mereadd:queries} shows some sample queries run on the scenario under the opt rule. On l.~26, the assumption of the transitivity of the betterness relation (on possible worlds) is introduced.  \emph{Sledgehammer} shows the inconsistency of (P0)-(P2). On l. 35, the assumption of transitivity is dropped. \emph{Nitpick} confirms the satisfiability of (P0)-(P2). The model generated by \emph{Nitpick} is shown in Fig.~\ref{fig:mereadd:henkin}. \newpage
\begin{figure}[h]
    \centering
    \includegraphics[scale=0.55]{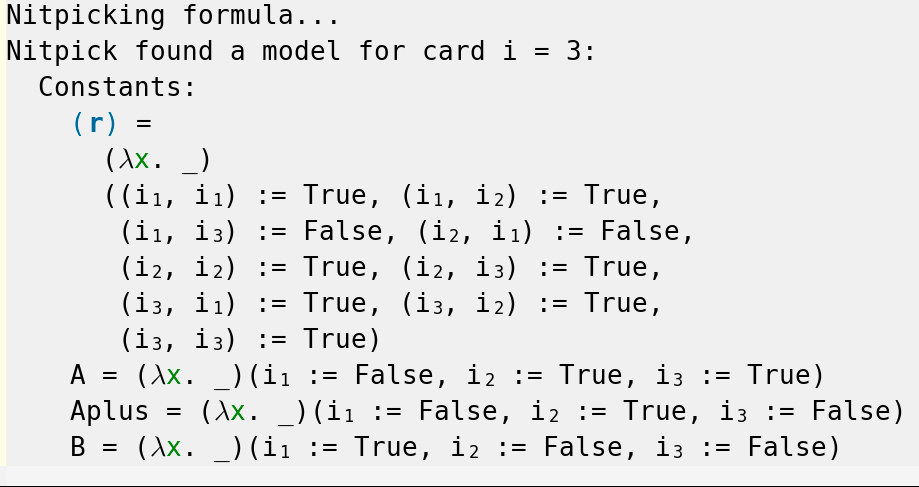}
    \caption{A non-transitive model satisfying (P0)-(P2)}
    \label{fig:mereadd:henkin}
\end{figure}

\begin{figure}[h]
    \centering
    \includegraphics[scale=0.5]{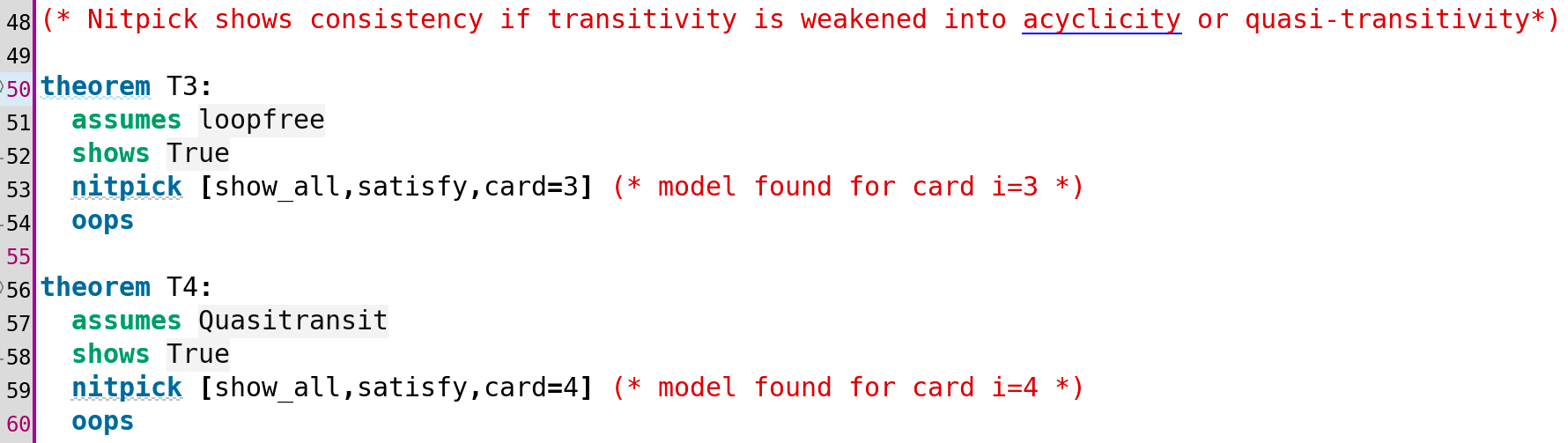}
    \caption{A-cyclicity and quasi-transitivity}
    \label{fig:pic/mere-acci-acycl}
\end{figure}\begin{figure}[h]
    \centering
    \includegraphics[scale=0.5]{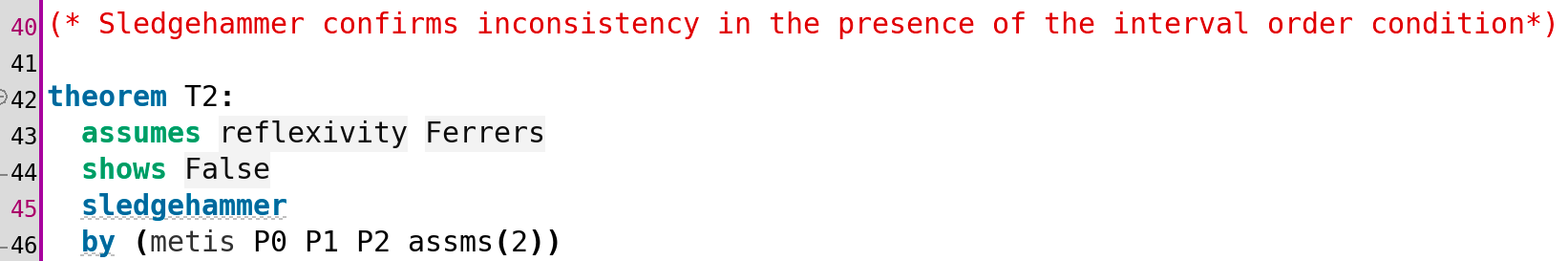}
    \caption{Interval order}
    \label{fig:pic/io-mere}
\end{figure}
 \emph{Nitpick} can also confirm that the mere addition paradox is avoided if transitivity is not rejected wholesale, but weakened into a-cyclicity or quasi-transitivity. This point has in general been overlooked in the literature. On the other hand, \emph{Sledgehammer} can verify that this solution does not work for the interval order condition, which represents another candidate weakening of transitivity. The verifications are shown in Figs. \ref{fig:pic/mere-acci-acycl} and  \ref{fig:pic/io-mere}.

\newnew{
We run the same queries under the max rule and the \ref{lewis} rule.  The findings are summarized in Table \ref{findings}. The left-most column shows the constraint put on the betterness relation. The other columns show what happens when varying the truth conditions for the conditional obligation operator. The symbol $\checkmark$ indicates that the sentences formalizing the scenario have been confirmed to be consistent, and the symbol  \ding{56} indicates they have been confirmed to be inconsistent. 
\begin{center}
\begin{table}[h]
\centering
\begin{tabular}{c|ccc}
\diagbox[trim=r, height=2\line]{\\ \\Property}{Truth conditions} &Opt&Max&\ref{lewis}\\ \hline
None & $\checkmark$ &  $\checkmark$&  $\checkmark$  \\
Transitivity + totality & \ding{56}& \ding{56}   & \ding{56} \\
Transitivity & \ding{56}& \ding{56} (if model finite)  & \ding{56} \\
Interval order   & \ding{56}& \ding{56} & \ding{56} \\
Quasi-transitivity   &  $\checkmark$ &  \ding{56}  (if model finite)  &  $\checkmark$   \\
Acyclicity& $\checkmark$ &  $\checkmark$&  $\checkmark$  \\
\hline
\end{tabular} \caption{Mere addition paradox (overview of findings)} \label{findings}
\end{table}
\end{center}
We verified manually the counter-models found by \emph{Nitpick}, and all appeared to be correct. One can see that changing the truth conditions for the conditional does not have any effect, except for transitivity and quasi-transitivity under the max rule. A few comments are in order. 

The formulas involved in the scenario are labeled as PP0-PP2. First of all, \emph{Sledgehammer} shows that, if $\succeq$ is transitive and total, then PP0-PP2 are inconsistent. This is shown in Fig. \ref{fig:pic/addtrtot}. 
\noindent
This is to be contrasted with the situation under the opt rule and the \ref{lewis} rule. But this apparent asymmetry has an explanation. 
When Temkin refers to the betterness relation, he has mostly in mind the relation $\geq$  (on formulas). He does not disentangle the relation $\geq$  (on formulas) from the relation $\succeq$ (on worlds), the latter being used to define the truth conditions of the former. Nor does he specify if ``best'' is to be understood in terms of maximality or optimality. 
The properties of $\succeq$ and $\geq$ may not coincide depending on the definition of ``best".
Thus, under the opt rule, if $\succeq$ is transitive, then $\geq$  is transitive $-$see T16 in Fig. \ref{transi-opt}. 
Under the max rule, it is only if $\succeq$ is \emph{both} transitive and total that $\geq$  is transitive$-$see T12 in Fig. \ref{totalplustr2}.\footnote{\newnew{A pen and paper verification may be found in \cite[Obs. 2.11 and 2.12]{ddl:P21}}.}

\begin{figure}[h]
    \centering
    \includegraphics[scale=0.25]{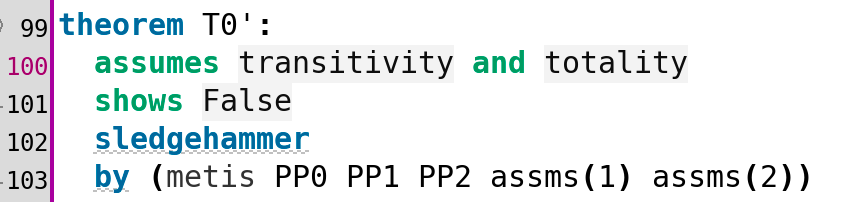}
    \caption{Inconsistency under transitivity and totality}
    \label{fig:pic/addtrtot}
\end{figure}
However, this does not explain everything. 
The above might suggest an alternative solution to the mere addition paradox. Perhaps one could just keep the transitivity of $\succeq$  but reject the totality of $\succeq$, while concurrently defining ``best" in terms of maximality. This would be in keeping with the conventional approach in rational choice theory: maximality is often deemed more suitable than optimality, because it keeps the possibility of incomparability open$-$\cite{ddl:S97}. 
But what happens with transitivity or quasi-transitivity of $\succeq$ alone (third and fifth row, starting from the top) suggests that the solution lies elsewhere. \emph{Sledgehammer} shows that, given (quasi-)transitivity, the formulas PP0-PP2 are inconsistent,  assuming a finite model
   of cardinality (up to) seven (if we provide the exact dependencies). This is shown in 
  Fig. \ref{fig:pic/addtr}.\footnote{\newnew{
For higher cardinalities 
   \emph{Sledgehammer} and  \emph{Nitpick} return a timeout.  Timeouts can be explicitly specified by a parameter, like [timeout=100]; the default is 60s. }}

\begin{figure}[h]
    \centering
    \includegraphics[scale=0.5]{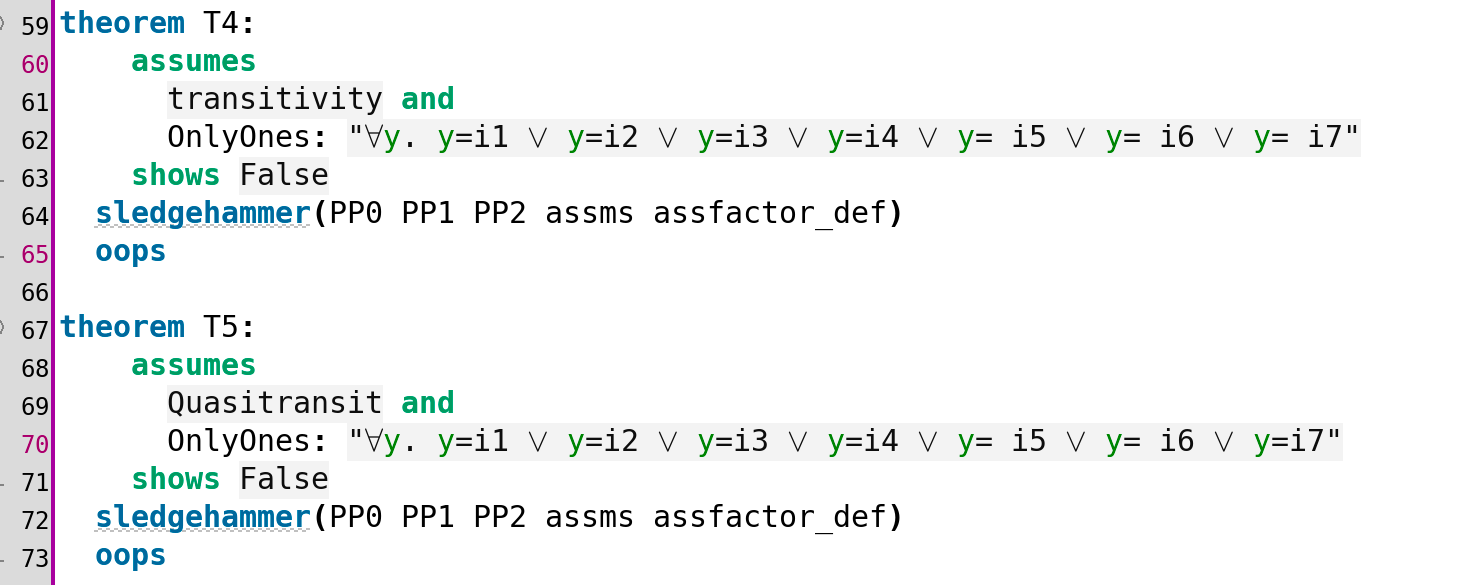}
    \caption{Inconsistency under (quasi-)transitivity and finiteness}
    \label{fig:pic/addtr}
\end{figure}
Thus, under the max rule, (quasi-)transitivity makes PP0-PP2 inconsistent if the set of possible worlds is assumed to be finite$-$an assumption that might appear overly limiting, if not arbitrary. This observation calls into question the idea that transitivity is the sole cause of the paradox.





The above fact has remained unnoticed until now.  We have not been able to establish it in full generality  (i.e., regardless of the model's fixed cardinality) without resorting to manual calculations. 
This is Proposition \ref{finiteness} below.   
It says that, in the presence of (quasi-)transitivity, a necessary condition for PP0-PP2 to be simultaneously satisfiable, is that the model contains 
an infinite increasing $\succ$-chain of $\neg A\wedge\neg A^{+}\!\!\wedge B$-worlds (all distinct). 
\begin{proposition}\label{finiteness} Suppose $\succeq$ is quasi-transitive (resp. transitive). Assume the following  formulas  are satisfied in a world in a model $M=(W,\succeq, v)$:
\begin{flalign}
&P(A/
A\vee B) \label{eq0}\\
&P(A^{+}\!\!/
A\vee A^{+}) \label{eq1}\\
&\bigcirc (\neg A^{+}\!\!/
A^{+}\!\!\vee B) \label{eq2}\\
&\bigcirc (\neg B/ A\vee B) \label{eq3}\\
& P(B/A^{+}\vee B) \label{eq4}
\end{flalign}
Then $W$ contains an infinite increasing $\succ$-chain of $\neg A\wedge\neg A^{+}\!\!\wedge B$-worlds (all distinct). 
\end{proposition}

\begin{proof} I focus on the case where $\succeq$ is quasi-transitive. Recall that by definition 
$\succ$ is irreflexive and asymmetric, and that quasi-transitivity entails acyclicity (see Fig. \ref{wt}). 
Assume that formulas (\ref{eq1})-(\ref{eq2})-(\ref{eq3}) are satisfied; the argument primarily revolves around these. 
Equations (\ref{eq0}) and (\ref{eq4}) can also be assumed true without leading to a contradiction.

By (\ref{eq1}), there is some $a_1$
such that $a_1\in\max (A\vee A^{+}\!)$ and $a_1\models A^{+}$. Hence, $a_1\models A^{+}\!\vee B$. By (\ref{eq2}),   $a_1\not\in\max (A^{+}\!\vee B)$, so there is some $a_2\succ a_1$ s.t.  $a_2\models A^{+}\vee B$. By irreflexivity, $a_1$ and  $a_2$ are distinct. Clearly,  $a_2\models \neg A\wedge\neg A^{+}$. So $a_2\models B$, and hence  by  (\ref{eq3}), 
 $a_2\not\in\max (A\vee B)$. It follows that there is some $a_3\succ a_2$ s.t.  $a_3\models A\vee B$.  By irreflexivity and asymmetry, $a_3$ is other than  $a_2$ and $a_1$. By quasi-transitivity, 
 $a_3\succ a_1$. Since  $a_1\in\max (A\vee A^{+}\!)$, $a_3\models \neg A\wedge\neg A^{+}$. Hence  $a_3\models B$, and hence  by  (\ref{eq3}) again, $a_3\not\in\max (A\vee B)$,
 and  so there is some  $a_4$ s.t. $a_4\models A\vee B$ and $a_4\succ a_3$.  By acyclicity, $a_4$ is other than  $a_1$, $a_2$, $a_3$ and $a_4$. By quasi-transitivity, $a_4\succ a_1$, and so as before $a_4\models \neg A\wedge\neg A^{+}$. Reiterating the above argument indefinitely, one gets an infinite increasing $\succ$-chain of $\neg A\wedge\neg A^{+}\!\!\wedge B$-worlds, starting with $a_2$. 

 The argument goes through if $\succeq$ is required to be transitive. This is because transitivity implies quasi-transitivity (see Fig. \ref{wt}). 
\end{proof}
\begin{remark}
    
Note that Prop. \ref{finiteness} does not apply to to the interval order condition (even if 
 this one implies quasi-transitivity as well). 
 This is because (\ref{eq4}) cannot simultaneously hold in a model where the interval order condition is satisfied.
 This can easily be verified. By (\ref{eq4}), there is some
 $b_1\in\max (A^{+}\!\vee B)$ with $b_1\models B$. By (\ref{eq3}),  $b_1\not\in\max (A\vee B)$. So there is some $b_2\succ b_1$ with $b_2\models A\vee B$. Clearly,
 $b_2\models A\wedge\neg B\wedge \neg A^{+}\!$. We have $a_1\not\succeq a_2$ and 
 $b_1\not\succeq b_2$. By Ferrers, $a_1\not\succeq b_2$ or  
 $b_1\not\succeq a_2$. By totality, $b_2\succ a_1$ or  
 $a_2\succ b_1$. The first contradicts the fact that 
$a_1\in\max (A\vee A^{+}\!)$, while the second contradicts the fact that $b_1\in\max (A^{+}\!\vee B)$. This is shown in Fig. \ref{fmp}, where the two cases are indicated by a ``or". 
\end{remark} 

 \begin{figure}[h]
    \centering
    \includegraphics[scale=0.35]{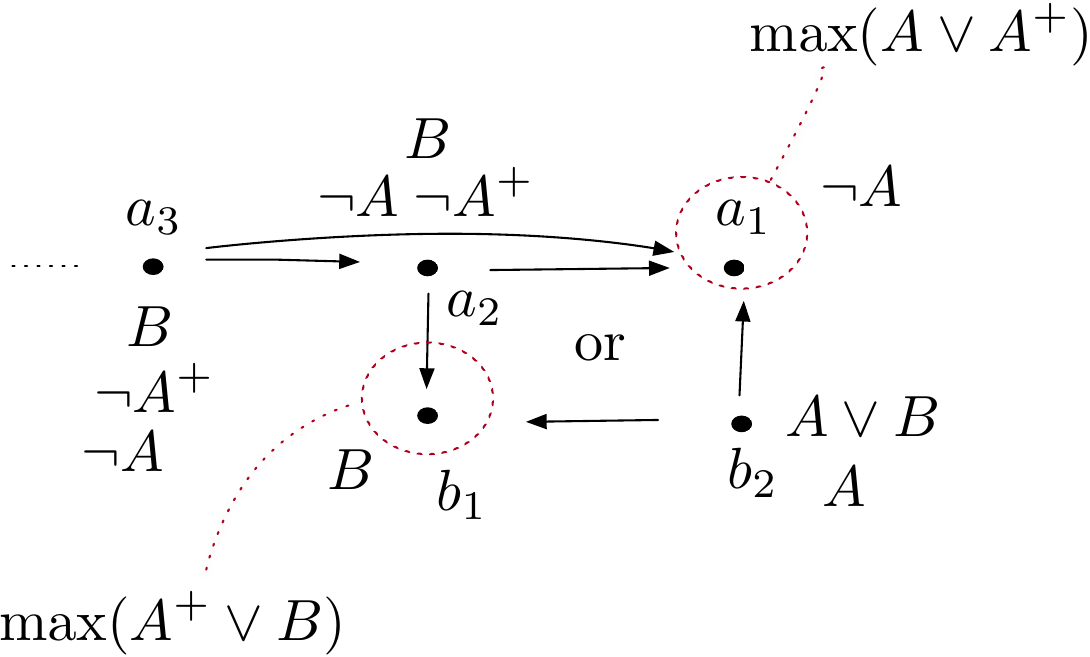}
    \caption{Adding (\ref{eq4})}
    \label{fmp}
\end{figure}

The following spin-off result is new to the literature. We recall that the finite model property (f.m.p.) is said to hold w.r.t. a given class $C$ of models, if any formula $\varphi$ that is satisfiable in class $C$ is satisfiable in a finite model in $C$.  
\begin{corollary}[f.m.p.]  Under the max rule, the finite model property fails w.r.t. the following classes of models whose relation $\succeq$ meets the property as indicated:
\begin{itemize}
\item $\succeq$ is quasi-transitive
\item $\succeq$ is transitive
\item $\succeq$ is an interval order
\end{itemize}
\end{corollary}
\begin{proof} The second and third claims follow from the first, because quasi-transitivity follows from  transitivity, and also from the interval order condition (see Fig. \ref{wt}). To prove the first claim, set $\varphi: = (\ref{eq1})\wedge (\ref{eq2})\wedge (\ref{eq3})$, and use Prop. \ref{finiteness} above. (The interval order condition will be met in the described model as long as each world is assumed to be at least as good as itself.) 
\end{proof}

 \begin{figure}[h]
    \centering
    \includegraphics[scale=0.53]{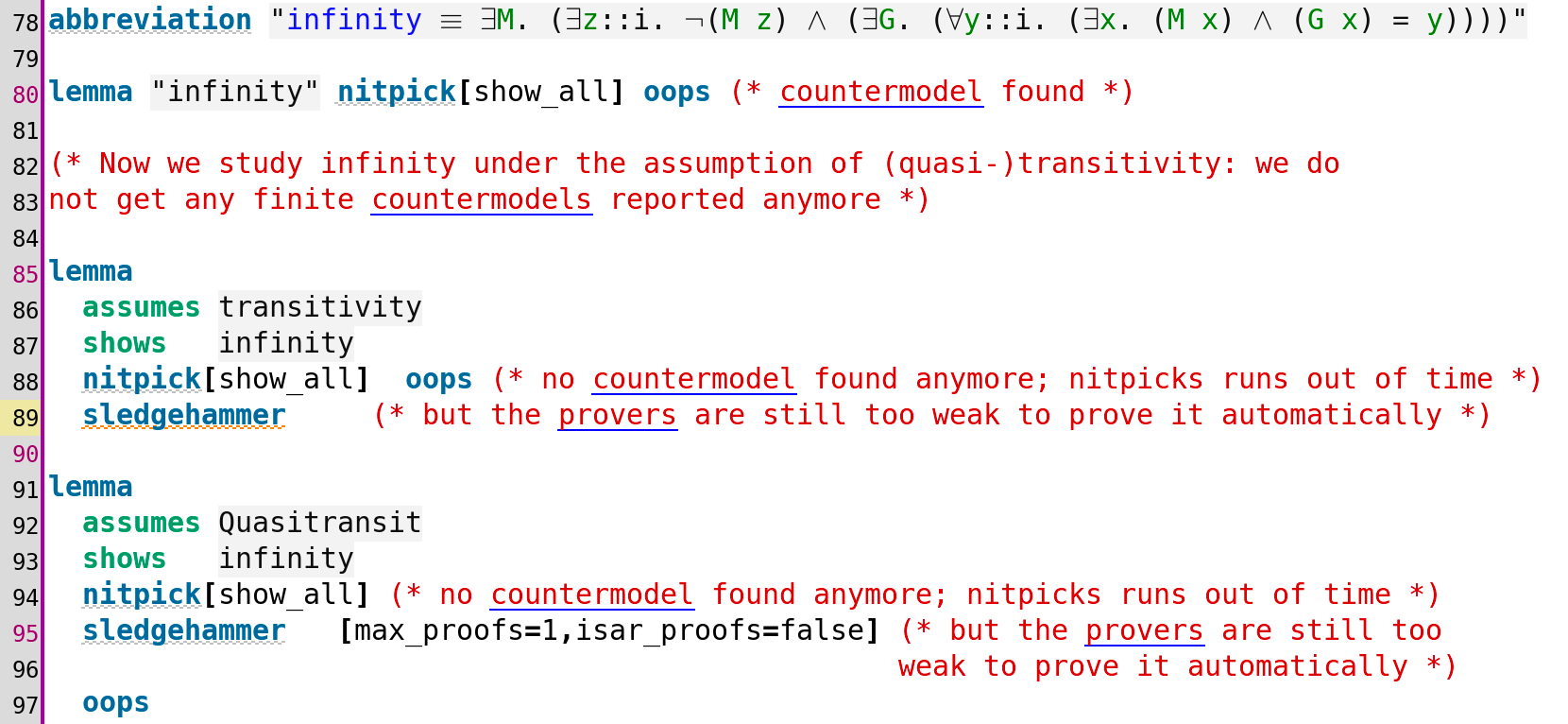}
    \caption{Proving infinity}
    \label{fig:pic/inf}
\end{figure}

In HOL one can define the axiom of infinity (for type $i$) by the second-order formula:  
   \begin{flalign*}
 \mbox{infinity} \equiv \exists M.\;\;\big( \exists z::i. \;\neg (Mz) \;\wedge\; (\exists G. \, (\forall y::i. \,(\exists x. \,(Mx) \wedge (Gx = y))))\big)
   \end{flalign*}
The \emph{definiens} says that there is a surjective mapping $G$ from  domain $i$ to a proper subset $M$ of domain $i$.  Testing whether infinity holds, \emph{Nitpick} gives us a counter-model to infinity that is a model of PP0-PP2.  If the same query is run under the assumption of (quasi-)transitivity, we do not get any (finite)
counter-model reported anymore. However the provers are still not strong enough to prove infinity.  This is shown in Fig.~\ref{fig:pic/inf}.
  
\subsection{Summary}\label{sum:f}

Distinguishing between ``better than'' as a relation between formulas and as a relation on possible worlds,  our formalisation offers two new insights on the scenario. Below, it is understood that the structural properties are those of the second relation. 
\begin{itemize}
\item One can choose not to take a stand on the truth conditions for the conditional, but weaken transitivity rather than reject it wholesale. However,  not all potential weakenings of transitivity prove effective: quasi-transitivity and acyclicity do the job, but not the interval order condition.  This is independent of the choice of the evaluation rule for the conditional. 
\item One could adopt the max rule, keep transitivity, and allow  for the possibility that there are infinite sequences of better and better worlds.
 Ultimately, this solution is questionable, for the reason explained in Sect. \ref{lew}: the max rule faces a deontic explosion problem, if infinite chains are allowed. Nevertheless, the availability of this option is worth a mention.
\end{itemize}

}
\section{Conclusion} \label{conc}

Utilising the LogiKEy methodology and framework we have developed mechanisations of extensions of \AA qvist's preference-based system  {\bf E} for conditional obligation. We have illustrated the use of the resulting tool for (i) meta-logical studies and for (ii) object-level application studies in normative reasoning. Novel contributions, partly contributed by the automated reasoning tools in Isabelle/HOL, include the automated verification of the correspondence between semantic properties and modal axioms, and the formalisation and mechanisation of Parfit's argument for the repugnant conclusion. This one reveals the possibility of a take on the scenario usually under-appreciated in the literature, which consists in weakening transitivity suitably. Future work includes the handling of the full equivalence between properties and formulas, the formalisation of (and comparison with) other solutions to the repugnant conclusion, and the analysis of other variant paradoxes discussed in the literature. 

\section*{Funding}

Dr. X. Parent was funded in whole, or in part, by the Austrian Science Fund (FWF) [M3240 N, ANCoR project (doi: 10.55776/I2982), and 10.55776/I6372, LoDEX project].
For the purpose of open access, the author has applied a CC BY public copyright licence to any Author Accepted Manuscript version arising from this submission.




\bibliography{interactapasample}
\bibliographystyle{apalike}



\end{document}